\documentclass[runningheads]{llncs}
\usepackage[letterpaper,top=2cm,bottom=2cm,left=3cm,right=3cm,marginparwidth=1.75cm]{geometry}

\usepackage{amsmath}
\usepackage{graphicx}
\usepackage{amssymb}
\usepackage[colorlinks=true, allcolors=blue]{hyperref}
\usepackage{algorithm}
\usepackage{algorithmic}
\usepackage{authblk}
\usepackage{makecell}
\usepackage{booktabs}
\usepackage{arydshln}
\newcommand{\ra}[1]{\renewcommand{\arraystretch}{#1}}
\usepackage{verbatim}
\usepackage{ulem}
\usepackage{graphicx}

\begin{document}
%section
\begin{comment}
\usepackage{titlesec}
\setcounter{secnumdepth}{4}
\titleformat{\paragraph}
{\normalfont\normalsize\bfseries}{\theparagraph}{1em}{}
\titlespacing*{\paragraph}
{0pt}{3.25ex plus 1ex minus .2ex}{1.5ex plus .2ex}    
\end{comment}

\title{Improved Privacy-Preserving PCA Using Optimized 
Homomorphic Matrix Multiplication}
\author{Xirong Ma\inst{1}\and Chuan Ma\inst{2}\and Yali Jiang \inst{1}\thanks{Corresponding author: Yali Jiang} \and Chunpeng Ge\inst{1,3}}
%\author[2]{Chuan Ma}
%\author[3]{Yali Jiang}
%\author[4]{Chunpeng Ge}
%\author{Xirong Ma\inst{1}}
\authorrunning{X. Ma et al.}
% \author[*]{Yali Jiang}
\institute{School of Software, Shandong University \and Zhejiang Lab \and Joint SDU-NTU Centre for Artificial Intelligence Research (C-FAIR)}
% \affil[*]{School of Software, Shandong University}

\maketitle

\begin{abstract}
Principal Component Analysis (PCA) is a pivotal technique widely utilized in the realms of machine learning and data analysis. It aims to reduce the dimensionality of a dataset while minimizing the loss of information. In recent years, there have been endeavors to utilize homomorphic encryption in privacy-preserving PCA algorithms for secure cloud computing. These approaches commonly employ a PCA routine known as PowerMethod, which takes the covariance matrix as input and generates an approximate eigenvector corresponding to the primary component of the dataset. However, their performance is constrained by the absence of an efficient homomorphic covariance matrix computation circuit and an accurate homomorphic vector normalization strategy in the PowerMethod algorithm. In this study, we propose a novel approach to privacy-preserving PCA that addresses these limitations, resulting in superior efficiency, accuracy, and scalability compared to previous approaches. %他们普遍使用一种称为 PowerMethod的 PCA routine， which takes the covariance matrix as input and produces an approximate eigenvector corresponding to the first principal component of the dataset. 然而他们的性能和准确率受限于无法高效地同态地计算出协方差矩阵，以及受限于没有一个通用的向量规范化策略来进行 PowerMethod 算法中的向量规范化。In this study, we present a novel approach for privacy-preserving PCA that tackles these limitations 而取得了比过往方案 \cite{panda2021principal} 更好的效率、准确性与可扩展性。

%However, their performance and accuracy are constrained by the incapability of homomorphic covariance matrix computation and the absence of a universal vector normalization strategy for the PowerMethod algorithm.

We attain such efficiency and precision through the following contributions: (i) We implement space and speed optimization techniques for a homomorphic matrix multiplication method, specifically tailored for parallel computing scenarios. (ii) Leveraging the benefits of this optimized matrix multiplication, we devise an efficient homomorphic circuit for computing the covariance matrix homomorphically. (iii) Utilizing the covariance matrix, we develop a novel and efficient homomorphic circuit for the PowerMethod that incorporates a universal homomorphic vector normalization strategy to enhance both its accuracy and practicality.
% (Jiang et al., SIGSAC 2018), 

Our privacy-preserving PCA scheme, implemented using our innovative homomorphic PowerMethod circuit, surpasses state-of-the-art approaches with an average speedup of 1.9 times on datasets with size $256\times 256$, along with an average improvement in R2 accuracy of $0.195$. Notably, our scheme demonstrates an even more remarkable estimated speedup of $25$ times when applied to larger datasets with size $60000\times 256$,  showcasing scalability that has not been reported by previous approaches.% Specifically, we successfully compute the top 8 principal components of the Fashion-MNIST dataset($60000\times 256$) in approximately 1 hour. Results obtained from such a dataset scale have not been reported by existing approaches.
% demonstrating an R2 accuracy of 0.9
% 在 $256\times 256$ 数据集上
%, and an average smaller error of 10 times between the encrypted and un-encrypted routines

% in terms of both scalability and accuracy. When compared to results obtained from datasets of similar scale, our scheme exhibits an average speed improvement of fourfold and achieves an average R2 accuracy improvement of 0.16. Additionally, the errors in our results, when compared to plaintext routines, are on average 10 times smaller. Furthermore, when applied to larger datasets, our scheme is estimated to achieve a speed improvement of up to 10 times. It is worth noting that our scheme successfully computes the top 8 principal components of the Fashion-MNIST dataset in approximately 1 hour, yielding an impressive R2 accuracy of 0.9. Such results have not been previously reported by existing approaches.

%Our privacy-preserving PCA scheme, based on our new homomorphic PowerMethod circuit, 在与 state of art 相同规模的result 比较时，我们的方案平均快 4 倍，并且 R2 准确率平均高 0.16，与明文routine结果的误差平均小 10 倍。在更大规模数据集上我们的方案甚至估算具有 10 倍更快的速度，Notably, our scheme successfully computes the top 8 principal components of the Fashion-MNIST dataset in approximately 1 hour, demonstrating an R2 accuracy of 0.9，这是过往方案所不曾报告过的效果。

% reported by Panda et al. (CSCML 21) 
%Our privacy-preserving PCA scheme based on our new homomorphic PowerMethod circuit 比汇报在 \cite{panda2021principal} 中的、当前最好的成绩平均快了 4 倍，并且有更高的准确率。我们的方案 successfully computes the top 8 principal components of the Fashion-MNIST dataset in approximately 1 hour, achieving an R2 accuracy of 0.9. 这样的可扩展性是未曾在过往研究中汇报过的。
\end{abstract}

\keywords{Privacy-preserving PCA\and CKKS \and Homomrophic matrix multiplication} \and Cloud computing \and Machine Learning as a service

\section{Introduction}

Principal Component Analysis (PCA) \cite{hotelling1933analysis}, \cite{pearson1901liii} is a widely employed dimensionality reduction technique. It maps high-dimensional data to a lower-dimensional space, preserving the most essential features while minimizing redundant information and noise. Specifically, PCA transforms the original data into a set of new orthogonal variables called principal components, representing the directions of maximum variance in the original data. PCA is widely adopted in data analysis and machine learning, including data preprocessing, feature extraction, data compression, and visualization. It helps us understand the relationships within a dataset and provides a simpler and more manageable representation. 

Due to the significance and versatility of PCA, it is leveraged as one of the supported technologies in cloud computing services, enabling users to perform enhanced data analysis leveraging the computational power of the cloud. However, concerns regarding the trustworthiness of cloud storage hinder users from directly performing analysis and processing of sensitive data in the cloud. Instead, it is needed to incorporate privacy protection measures into their data before outsourcing it to mitigate the risk of potential attacks from an untrusted cloud environment. Specifically, users and cloud servers can negotiate and employ privacy-preserving algorithm routines to process sensitive data in the cloud while ensuring privacy. Therefore, the proposition of a practical and efficient privacy-preserving PCA technique would greatly benefit users by enabling them to perform PCA on sensitive data in the cloud. %具体来说，用户和云服务器可以通过协商，采用隐私保护的算法例程来在云上处理敏感数据。那么 the proposition of a practical and efficient privacy-preserving PCA technique 将有利于用户在云上对敏感数据进行 PCA。

One possible approach for implementing such a privacy-preserving PCA technique involves utilizing homomorphic encryption. Homomorphic encryption is an important privacy-preserving technique that allows computations to be performed on encrypted data without decryption. This technology serves to safeguard data privacy while preserving data usability, rendering it extensively employed in domains such as cloud computing and cross-domain computation. Over the past fifteen years, a class of homomorphic schemes based on the Ring Learning With Errors (R-LWE) problem \cite{lyubashevsky2010ideal} has rapidly developed (e.g., but not limited to \cite{brakerski2014leveled}, \cite{fan2012somewhat}, \cite{cheon2017homomorphic}, \cite{cheon2019full}). These schemes naturally possess SIMD properties and support homomorphic addition and multiplication operations. As a result, a plethora of privacy-preserving data analysis algorithms based on these schemes have emerged, including privacy-preserving PCA schemes tailored for cloud service scenarios \cite{lu2016using}, \cite{rathee2018faster}, \cite{panda2021principal}. These existing homomorphic encryption-based privacy-preserving PCA methods utilized an iterative algorithm known as PowerMethod to compute the dominant eigenvector of the covariance matrix of the dataset, which corresponds to the first principal component. This algorithm selects an initial approximation of the dominant eigenvector and continually applies the covariance matrix transformation to refine its approximation (see Algorithm \ref{alg1} for detail). There are two major limitations among the homomorphic PoweMethod algorithms in the previous approaches. 

Firstly, these approaches lack the capability to compute the covariance matrix homomorphically. the PowerMethod necessitates the input of the dataset covariance matrix, but to our knowledge, previous schemes have not offered a homomorphic solution for computing this matrix. Instead, they resort to alternative methods. Some require users to compute the covariance matrix locally which burdens the users with additional computational tasks and deviates from the original intent of cloud services \cite{rathee2018faster}. Others decompose the covariance matrix transformation within the PowerMethod into dataset matrix transformations to avoid explicit involvement of the covariance matrix but introduce extra computational complexity\cite{panda2021principal}.

Secondly, there is a potential loss of accuracy due to the absence of a universal vector normalization strategy. In each iteration of PowerMethod, normalization is required to control the length of the vector. In the homomorphic context, iterative algorithms are typically used to approximate the inverse square root function for normalization. The accuracy of these algorithms heavily relies on the selection of parameters such as the evaluation interval and the number of iterations. To our knowledge, no prior research has presented a universal strategy to determine these parameter settings in the realm of homomorphic PowerMethod algorithms. Consequently, existing homomorphic PowerMethod algorithms may potentially suffer from inaccuracy due to this inherent limitation.

\subsection{Our Contributions}
We propose an efficient privacy-preserving PCA algorithm that overcomes the obstacles above. To achieve this, we make the following contributions.
\newline
\noindent \textbf{Homomorphic Covariance Matrix Computation with Optimized Matrix Multiplication (Section 4,5):} We enhance the efficiency of a cutting-edge homomorphic matrix multiplication algorithm, both in terms of time and space, to render it more apt for parallel computing scenarios. Subsequently, we employ it as a fundamental element to design a proficient algorithm for homomorphic covariance matrix computation, which harnesses the parallel computation capabilities of multiple matrix multiplication instances.
%and use it as a core component to design an efficient homomorphic covariance matrix computation circuit that takes advantages of the parallel computation of multiple instances of matrix multiplication. 

%\noindent \textbf{Improved Matrix Multiplication (Section 4)} In order to compute the covariance matrix homomorphically, we first explore the use of a state-of-art homomorphic matrix multiplication algorithm combined with some existing optimization techniques. However, we encountered memory saturation due to the large rotation key space required by the algorithm when attempting to perform parallel matrix multiplications for large datasets. This motivated us to seek ways to reduce the rotation key space for matrix multiplication while maintaining or even improving its computational efficiency.

%\noindent \textbf{Homomorphic Covariance Matrix Computation (Section 5)} Leveraging the improved homomorphic matrix multiplication, we design a circuit for homomorphically computing the covariance matrix of large datasets.  It takes advantage of parallel computation using multiple instances of matrix multiplication and exploits the symmetry of the covariance matrix to enhance efficiency.%借助改良搞的同态矩阵乘法，我们为同态计算大数据集的协方差矩阵设计了电路，它利用多矩阵乘法实例的并行计算以及协方差矩阵的对称性来提速。

\noindent \textbf{Privacy-preserving PCA using Homorphic PowerMethod (Section 6):} PowerMethod consists of two main components: covariance matrix transformation and vector normalization. We first design an efficient homomorphic circuit for covariance matrix transformation. Next, we employ an iterative algorithm approximating the inverse square root (InvSRT) function to perform the vector normalization and propose a systematic approach to parameterize the iterative algorithm. Finally, we present the process of performing privacy-preserving PCA utilizing our homomorphic PowerMethod circuit. %PowerMethod 主要由两个部分：协方差矩阵变换和向量规范化。我们为协方差矩阵变换设计了高效的计算电路，然后使用方根求逆的迭代算法对向量进行规范化，我们详细地讨论并给出迭代算法的求值区间和迭代轮数的选取方法，因此使 PowerMethod 在同态计算下有更精确的结果。discuss in detail the selection of evaluation range and number of iterations for the iterative algorithm,

% ---------------------------------- 2023/5/28 -----------------------------------------------

\noindent \textbf{Implementation (Section 7)} We compare our proposed solution with existing approaches, highlighting that our algorithm outperforms previous methods in terms of computational efficiency, precision, and scalability. %我们将我们的方案和过往的方案进行性能上的比较，并指出我们的算法在计算效率、精确度和可扩展性上都要优于过往的方案。

\section{Related Work}

\subsection{Privacy-preserving PCA in cloud computing scenario}

The initial attempts at performing PCA on an RLWE-based homomorphic encryption scheme are made by Lu \textit{et al.} and Rathee \textit{et al.} \cite{lu2016using}, \cite{rathee2018faster}. They propose using an iterative algorithm called PowerMethod to compute principal components in a homomorphic encryption setting. The PowerMethod iteratively computes $\mathbf{v} \leftarrow \text{Cov}\cdot \mathbf{v}$, and it can be proven that $\mathbf{v}$ converges to the dominant eigenvector of the covariance matrix $\text{Cov}$ after a finite number of steps, where the dominant eigenvector is equivalent to a principal component. However, they find the homomorphic computation of the covariance matrix to be inefficient for arbitrary sample sizes. Therefore, their scheme requires the client (of cloud service) to compute and send the encryption of $\sum \mathbf{x}_i^T \mathbf{x}_i$ (where $\mathbf{x}_i$ represents the row vector form of the samples in the dataset $X$), which not only causes the computational burden on the user but also becomes trivial because PowerMethod is no harder than computing $\sum \mathbf{x}_i^T \mathbf{x}_i$ in plaintext space. Furthermore, their PCA schemes cannot normalize the approximate eigenvector $\mathbf{v}$ during the iterations of PowerMethod, and the length of the approximate eigenvector is prone to overflow since the homomorphic encryption they used only supports homomorphic operations on modular integers \cite{brakerski2014leveled}.

Subsequently, A work by Panda \cite{panda2021principal} migrates the PowerMethod to another RLWE-based homomorphic encryption scheme known as CKKS that supports approximate computation on complex numbers \cite{cheon2017homomorphic}. The CKKS scheme enables them to perform vector normalization using an iterative algorithm of the inverse square-root function. However, their work provides only one possible implementation of the iterative algorithm without demonstrating how to choose its evaluation interval and the number of iterations. As a result, their solution lacks both accuracy and practicality. Furthermore, The PowerMethod in their work is done by iteratively computing $\mathbf{v}\leftarrow 1/N\cdot X^T(X\mathbf{v})$ (where $X$ is the centered dataset), rather than directly computing the covariance matrix transformation. This makes the complexity of the PowerMethod dependent on the size of the dataset rather than the size of the covariance matrix, which is inefficient for datasets with a large number of samples. 

% In addition, The aforementioned approaches lack discussions on scalability. They conduct experiments using datasets with a number of units less than two hundred thousand, without considering larger datasets with tens of millions of units. However, datasets of this scale are common in real-world scenarios. This motivates us to explore homomorphic PCA schemes that are suitable for larger-scale datasets.

\subsection{Homomorphic Encryption Matrix Multiplication}
Numerous approaches are proposed on how to construct secure matrix multiplication algorithms using homomorphic encryption schemes based on the RLWE problem \cite{duong2016efficient}, \cite{mishra2021fast}, \cite{rathee2018faster}, \cite{jiang2018secure}, \cite{wang2019secure}. Among them, the work by Jiang \textit{et al.} \cite{jiang2018secure} introduces a scheme for secure real number matrix multiplication based on the CKKS scheme. To the best of our knowledge, it remains the state-of-the-art RLWE homomorphic encryption-based approach for real number matrix multiplication. In a recent study by Jang \textit{et al}. \cite{jang2022privacy}, they further improve this scheme by migrating it to a variant of CKKS based on the multivariate polynomial learning with errors (m-RLWE) problem. 

Nevertheless, few have systematically discussed the performance of homomorphic matrix multiplication within parallel computation scenarios. We believe that delving into this matter is of utmost importance, as parallel computation plays a pivotal role in expediting extensive matrix operations that involve multiple ciphertexts. Efficiently executing multiple instances of homomorphic matrix multiplication in parallel is instrumental in accelerating large-scale computations, warranting a thorough exploration of this aspect.

\section{Preliminaries}

\subsection{General Notation}

We use italic letters, such as $a$, to represent polynomial elements or numbers. We use bold letters, such as $\mathbf{a}$, to represent vectors, and uppercase letters, such as $A$, to represent matrices. The symbols $\oplus$ and $\odot$ denote elementwise addition and multiplication. The notation $\rho(\mathbf{a};k)$ represents the result of cyclically left-shifting (or rotating) the components of vector $\mathbf{a}$ by $k$ positions. Additionally, we use $(\mathbb{Z}/q\mathbb{Z},+)$ to represent a modulo $q$ additive group, where we use two sets of integers: $(-q/2,q/2]$ and $[0,q-1]$ to refer to two different representations of the residue class. We define $[a]_q$ as the result of reducing $a$ modulo $q$ that falls into $(-q/2,q/2]$, and $a \mod{q}$ as the result falling into $[0,q-1]$.

\subsection{Homomorphic Encryption for Arithmetic of Approximate Numbers}

Homomorphic encryption is a cutting-edge field in cryptography and serves as a powerful component for privacy-preserving computations. 
The CKKS scheme we mentioned earlier is a homomorphic encryption scheme supporting approximate arithmetic on complex vector space. It provides an encoding method to store up to $N/2$ floating point values in plaintext or ciphertext, which are represented by polynomials in the domain $R_{Q_L} = \mathbb{Z}_{Q_L}[X]/(X^N+1)$, where $Q_L=\prod_0^L q_i$, and $N$ is a power of 2. Operations over the domain are performed simultaneously on all encoded values, known as SIMD processing. The security of CKKS is based on the RLWE problem, so some noise is introduced into the encrypted values. We briefly introduce the operations provided by CKKS in the following.
\begin{itemize}
    \item $s\leftarrow \text{SecKeyGen}(\cdot)$ generates the secret key $s$
    \item $pk,swks\leftarrow \text{PubKeyGen}(s)$ generates public key $pk$ and a series of switching keys $swks=\{swk_{s_i,s_j}\}$ where $swk_{s_i,s_j}$ helps switch the secret inside a ciphertext from $s_j$ to $s_i$. Notice that switching keys can be quite large in size. Typically, a group of switching keys together can occupy a space ranging from tens of megabytes to several gigabytes or even more.
    \item $ct(m)\in R_{Q_L}^2\leftarrow \text{Enc}(pk,m)$ encryptes a plaintext vector $p$. 
\item $m'\leftarrow \text{Dec}(sk,ct(m))$ decrypts the ciphertext $ct(m)$ and output $m'\approx m$ 
\item $ct(m_1+m_2)\leftarrow \text{Add}(ct(m_1),ct(m_2))$ performs addition of the encrypted vectors.
\item $ct(m_1\odot m_2)\leftarrow \text{PMult}(ct(m_1),m_2)$ performs elementwise multiplication between encrypted vector $ct(m_1)$ and plaintext vector $m_2$. Rescaling is needed to control the noise growth in the result.
\item $ct(m_1\odot m_2)\leftarrow \text{Mult}(ct(m_1),ct(m_2))$ performs elementwise multiplication between encrypted vectors. Rescaling and relinearization operations are required to control the size and noise growth of the result. 
\item $ct(\rho(m,k))\leftarrow \text{Rot}(ct(m),rtk_k)$ performs cyclic rotation of step $k$ on the encrypted vector, where $rtk_k$ is an instance of switching key $swk_{s^{5^k},s}$ and $s$ is the secret inside the encrypted vector.
\end{itemize}

Notably, for a given set of modulus products $Q_L$, the homomorphic operations are limited to a multiplication depth of at most $L$ because rescaling consumes the modulus level of ciphertexts ( any ciphertext in $R_{Q_l}^2$ is at the $l$-th modulus level, and rescaling brings the ciphertext to a lower level $R_{Q_{l-1}}^2$). Once the modulus level of a ciphertext reaches 0, it can no longer participate in homomorphic multiplication operations. To perform operations with a greater multiplication depth, we need to conduct an operation called modulus refresh which raises the modulus level of the ciphertext. The modulus refresh can be achieved either by decrypting and re-encrypting the ciphertexts or by applying the bootstrapping function to the ciphertexts. 

We build our privacy-preserving PCA scheme upon a full Residue Number System (RNS) variant of the CKKS scheme \cite{cheon2019full} implemented in the Lattigo library \cite{lattigo}. This variant allows efficient polynomial multiplication by representing polynomials in both the RNS and NTT (Number Theoretic Transform) domains. We refer to Appendix \ref{apd:A} for more details of the full-RNS CKKS scheme.

% In summary, the full-RNS CKKS scheme enables us to perform Hadamard addition and multiplication on ciphertexts, as well as rotations of the message vector components. Our subsequent design of homomorphic algorithms will solely rely on the routines provided by the full-RNS CKKS scheme, ensuring that our algorithms are feasible within this framework. We use the parameter set  $\log{N},\log{QP},\sigma,\log{\text{Slots}}$ to identify a concrete CKKS instance, where $N,\log{QP},\sigma$ corresponds to the parameters $N,b,\sigma$ in function Setparams, and $Slots$ corresponds to the parameter $n$ in function Encode. The plaintext space of this instance is set as $\mathbb{C}^\text{Slots}$. %总的来说，full-RNS CKKS 方案实现了密文上的 SIMD 加法与乘法，以及向量分量的旋转。我们后续的同态算法设计只会使用上述 full-RNS CKKS 方案提供的 routine，使得我们的算法在该方案上是可实现的。

% The introduction of full-RNS CKKS can be divided into the following four modules. \\ %the general module, addition module, multiplication module, and rotation module. \\

% 2* ( (n_1+n_2')(MS+Pm) + (n_2'+1)(Dp+MD) + n/(n_1*n_2')(MS+Pm+Dp+MD) )

\subsection{Homomorphic Matrix Operation}
Our optimization for the homomorphic matrix multiplication is derived from the scheme proposed by Jiang \textit{et al.} \cite{jiang2018secure} for the CKKS scheme. They utilize a homomorphic linear transformation technique to build up the homomorphic matrix multiplication.  We will introduce both of these techniques in the following. 
 %以下会先对该矩阵乘法中用到的同态线性变换技术进行介绍，然后对矩阵乘法进行阐述。

\subsubsection{Homomorphic Linear Transformation}

Halevi \textit{et al.} first propose an approach to achieve linear transformations in the context of homomorphic encryption \cite{halevi2014algorithms}. They pointed out that any linear transformation $U\mathbf{m}$ can be represented as:
\begin{align}\label{eq2}
U\cdot \mathbf{m} = \sum_{0\leq \ell< n}(\mathbf{u}_{\ell}\odot\rho(\mathbf{m};\ell)),
\end{align}
where $\mathbf{u}_\ell$ represents the $\ell$-th diagonal vector of $U$: $\mathbf{u}_\ell=(U_{0,\ell},U_{1,\ell+1},...,U_{n-1,(\ell+n-1)\mod{n}})$. By associating $\mathbf{m}$ with the ciphertext and $U$ with the plaintext matrix in (\ref{eq2}),  we can effectuate a homomorphic linear transformation $U\mathbf{m}$ through $n$ ciphertext rotations, $n$ plaintext multiplication, and $n-1$ ciphertext additions. An algorithm known as Baby Step Giant Step (BSGS) can be employed to minimize the number of rotation operations involved \cite{halevi2018faster}: if there exists an arithmetic progression $\{s_k= a\cdot k |-d< k <d,d=d_1d_2\}$ such that all indices $\ell$ corresponding to non-zero diagonal vectors in $U$ belong to this arithmetic progression, then Equation (\ref{eq2}) can be reformulated as follows:

\begin{align}\label{eq3'} 
\begin{split}
U \mathbf{m}  &=  \sum_{-d_2< i<d_2,0\leq j< d_1 } \left( \mathbf{u}_{a\cdot(d_1\cdot i+j)} \odot \rho(\mathbf{m};a\cdot (d_1\cdot i+j))    \right)\\ 
&= \sum_{-d_2< i < d_2} \rho \left(\sum_{0\leq j < d_1} \left( \rho(\mathbf{u}_{a\cdot (d_1\cdot i+j)};-a\cdot d_1\cdot i)\odot \rho(\mathbf{m};a\cdot j)\right) ; a\cdot d_1\cdot i \right).
\end{split}
\end{align} 

Further optimization can be achieved by analyzing the rotation operations. Halevi \textit{et al.} introduce a ciphertext rotation optimization technique called \textit{hoisting} which reconstructs the internal operations for rotating a ciphertext $c=(c_0,c_1)\in R_{Q_l}^2$ by $k$ steps in the following sequence. (i) \textbf{Decompose}: Decompose $c_1$ into a vector $\mathbf{c}$ based on the decomposition base $\mathbf{b}$ of the rotation key $rtk_k$. (ii)\textbf{Permute}: Perform automorphism $\phi_k:X\rightarrow X^{5^k}(\mod{X^N+1})$ on each component of $\mathbf{c}$. (iii)\textbf{MultSum}: Perform the inner product $(c_0',c_1')\leftarrow \left<\mathbf{c}, \text{rtk}_k\right>$, resulting in a \textit{hoisted} ciphertext in $R_{Q_lP}^2$. (iv) \textbf{ModDown}: Reduce the modulus of the hoisted ciphertext to bring it back to $R^2_{Q_l}$. Here, \textbf{Decompose} and \textbf{ModDown} involve NTT and CRT basis extensions, which dominate the computational complexity of the rotation operation. The rotation key is only involved in the \textbf{MultSum}. Therefore, multiple rotations performed on the same ciphertext can share the result of \textbf{Decompose}(ct) when all the corresponding rotation keys are available. 
%This technique 将对在 $R_{Q_l}$ 中的密文 $c$ 旋转 $k$ 步的内部操作重构为以下顺序：(i) $c$ 按照分解基 $\mahtbf{b}$ 被分解成 $\beta$ 个分量的向量 $\mathbf{c}$ (ii) 对 $\mathbf{c}$ 中的每个分量进行排列 (iii) 进行内积：$\left<\mathbf{c},rtk_k\right> $ 此时得到的是属于 $R_{Q_lP}$ 的 hoisted 的密文(iv) 将 hoisted 密文的模数降低，使其回到 $R_{Q_l}$。

Bossuat \textit{et al.} apply the hoisting idea to the BSGS-optimized linear transformation algorithm \cite{bossuat2021efficient}. They apply the hoisting technique to the rotations $\rho(\mathbf{m};a\cdot j),0\leq j \leq d_1$ in Equation (\ref{eq3'}) but do not perform the final step of the rotation, \textbf{ModDown}, immediately. Instead, they multiply the intermediate hoisted ciphertexts with the corresponding pre-rotated diagonal vectors $\rho(\mathbf{u}_{a\cdot (d_1\cdot i +j)}; -a\cdot d_1\cdot i)$ and aggregate the products in $R^2_{PQ_\ell}$. \textbf{ModDown} is finally applied to the aggregated result. This optimization method is referred to as the \textit{double-hoisting} technique. It reduces the computational complexity of the BSGS-optimized linear transformation from $(d_2+d_1)\cdot (\textbf{MultSum} + \textbf{ModDown} + \textbf{Decompose} + \textbf{Permute})$ to $(d_2+d_1) \cdot (\textbf{MultSum} + \textbf{Permute}) + (d_2+1) \cdot (\textbf{Decompose} + \textbf{ModDown})$. 

% 这里我们还想讨论一下 double hoisting 技术加持下的同态线性变换的空间复杂度，它与我们后续将要讨论的同态矩阵乘法的并行运算的空间复杂度有紧密的联系。The space complexity comprises two primary components: (i) 需要参与密文旋转的 $d_2+d_1-2$ 把 rotation keys 以及 (ii) the $d_1$ intermediate hoisted ciphertexts for the BSGS inner loop. hoisted ciphertexts 占据的空间 is much smaller than rotation keys 占据的空间，就单个线性变换实例而言。然而如果我们考虑多个相同的线性变换实例同时进行，那么 rotation keys 是可以在这些实例之间共享的，但 hoisted ciphertexts 则会随并行实例增多而增多，并可能占据与旋转密钥相同量级的空间。这种场景将出现在矩阵乘法的并行运算中，因为线性变换是矩阵乘法的重要构成部分。
%那么结合上述 double hoisting 加持的线性变换的计算复杂度，我们可知 a smaller inner loop results in fewer pre-rotated ciphertexts but a slower speed, and a larger inner loop leads to more extensive use of pre-rotated ciphertexts but faster speed. Meanwhile, the number of rotation keys is minimized when the inner and outer loops are closest to each other.  

%For the double hoisting technique, the choice of inner loop counts $d_1$ can be adjusted to balance space and computation complexity. The space complexity comprises two primary components: the $d_1$ pre-rotated ciphertexts for the BSGS inner loop and the rotation keys that can be shared among instances of the same linear transformation. It is evident that a smaller inner loop results in fewer pre-rotated ciphertexts but a slower speed, and a larger inner loop leads to more extensive use of pre-rotated ciphertexts but faster speed. Meanwhile, the number of rotation keys is minimized when the inner and outer loops are closest to each other. 

%我们可以看出，较小的内层循环意味着较小的预计算密文数与较慢的速度，反之则预计算密文占用空间过多，速度较快，而旋转密钥数总是在内外层循环最接近时最小。

In the subsequent discussion, we assume that all the rotation keys required for performing any homomorphic encryption linear transformation are loaded into memory prior to the computation. This is the current implementation strategy in the Lattigo library. Although it is theoretically possible to load a key into memory only when it is needed for computation and release it soon afterward to control a smaller memory footprint for keys, we do not recommend this practice. Because it may incur additional time for disk read/write operations and disrupt the compactness of the storage structure. Especially, it will significantly slow down the speed of consecutive linear transformation computation. %Moreover, when considering parallel computation of multiple linear transformations, controlling the execution order of threads and the access order to temporary retrieved keys would become increasingly difficult and complex, and the space benefits obtained would not compensate for the loss in computational speed.

\subsubsection{Hommorphic Matrix Multiplication by Jiang \textit{et al.}}

Jiang \textit{et al.} proposed a homomorphic matrix multiplication approach based on the aforementioned linear transformation algorithm and its BSGS optimization \cite{jiang2018secure}. A square matrix $A$ of size $n\times n$ can be encoded as a vector $\mathbf{a}$ with $n^2$ components using the row ordering encoding.  Considering the multiplication of matrices $A$ and $B$, Jiang \textit{et al.} first encode $A,B$ into vectors $\mathbf{a}$ $\mathbf{b}$ with row ordering, then compute the following equation:
% apply a series of different linear transformations to $\mathbf{a}$ and $\mathbf{b}$, then use coordinate-wise multiplication to pairwise multiply the transformed results and aggregate them to obtain the encoding vector $\mathbf{c}$ of the matrix $AB$ as follows:
%For simplicity, let $n^2$ be equal to the number of available slots in a ciphertext (recall that $N/2$ is the maximum number of available slots, where $N$ is the degree of the cyclotomic polynomial $\Phi_M(X)$).
\begin{align}\label{eq4}
\mathbf{c}=\sum_{k=0}^{n-1}(C^kZ\mathbf{a})\odot(R^kT\mathbf{b}),
\end{align}
where $C$, $Z$, $R$, and $T$ are specific permutation matrices, and the $\mathbf{c}$ is the encoding vector of the matrix $AB$. They also design another permutation matrix transformation $G$ for matrix transposition: $\mathbf{a}' = G\mathbf{a}$, such that $\mathbf{a}'$ represents the transpose $A^T$ of the matrix $A$ represented by $\mathbf{a}$. The BSGS algorithm is employed for the transformations $Z, T$ and $G$ to greatly enhance efficiency. However, Jiang \textit{et al.} treat BSGS as a black-box operation without considering internal optimization by hoisting techniques. 

\subsection{Principal Component Analysis using PowerMethod}

Principal Component Analysis (PCA) is frequently employed for the dimensionality reduction of a dataset. Its fundamental concept involves identifying orthogonal axes that maximize the variance of the dataset projected onto these axes, which are referred to as principal components. Finding the top $k$ principal components is equivalent to determining the $k$ largest eigenvectors of the covariance matrix of the dataset.

% Recall figure \ref{fig:intro}, 
The covariance matrix $\Sigma$ of the dataset $X$ is given by $\Sigma = \frac{1}{N}X^TX - \mu \mu^T$, where $u^T = \frac{1}{N}\sum_{i=0}^{N-1}\mathbf{x}_i^T$. The PowerMethod is an approximate algorithm for computing the dominant eigenvector of the covariance matrix (see Algorithm \ref{alg1}). It accepts a covariance matrix and its initial approximation of the dominant eigenvector and continuously applies the covariance matrix transformation to the approximate eigenvector to make it more and more accurate. The overflow should be prevented by normalizing the approximate vector after each matrix transformation. The EigenShift algorithm is used to shift the covariance matrix in terms of eigenvectors (see Algorithm \ref{alg2}). It takes the covariance matrix and the top $k$ eigenvectors as inputs and outputs the $k$-shifted covariance matrix, where the $(k+1)$-th dominant eigenvector of the original covariance matrix becomes the $1$-th dominant eigenvector of the shifted one. By combining the PowerMethod and the EigenShift, we can compute the top principal components of the dataset. 
\begin{algorithm}
	%\textsl{}\setstretch{1.8}
	\renewcommand{\algorithmicrequire}{\textbf{Input:}}
	\renewcommand{\algorithmicensure}{\textbf{Output:}}
	\caption{PowerMethod}
	\label{alg1}
	\begin{algorithmic}[1]
        \REQUIRE
            $\Sigma$: covariance matrix of the dataset;
            $l_P$: number of iterations.
        \ENSURE
            $\mathbf{u}_1,\lambda_1$: dominant eigen-vector of $\Sigma$ and its eigen-value
		\STATE Choose a random vector $\mathbf{v}^{(0)}$ of size $d$
        \FOR{$i=1$ to $l_P$}
            \STATE $\mathbf{v}^{(i)}\leftarrow \Sigma\mathbf{v}^{(i-1)}$
            \STATE $\mathbf{v}^{(i)}\leftarrow \mathbf{v}^{(i)}/||\mathbf{v}^{(i)}||$
        \ENDFOR
        \RETURN $\mathbf{u}_1 = \mathbf{v}^{(l_P)}$ and $\lambda_1 = ||\mathbf{v}^{(l_P)}||/||\mathbf{v}^{(l_P-1)}||$ 
	\end{algorithmic}  
\end{algorithm}

\begin{algorithm}
    \renewcommand{\algorithmicrequire}{\textbf{Input:}}
    \renewcommand{\algorithmicensure}{\textbf{Output:}}
    \caption{EigenShift}
    \label{alg2}
    \begin{algorithmic}[1]
    \REQUIRE
        $\Sigma$: covariance matrix of the dataset;
        $\{\mathbf{u}_i,\lambda_i\}$: the $i$-th dominant eigenvector of $\Sigma$ and its eigenvalue.
    \ENSURE
        $\Sigma_k$: $k$-shifted covariance matrix of $X$
    \FOR{$i=1$ to $k-1$}
        \STATE $\Sigma_{i+1} = \Sigma_{i} - \lambda_i \mathbf{u}_i\mathbf{u}_i^T$
    \ENDFOR
    \RETURN $\Sigma_k$    
    \end{algorithmic}
\end{algorithm}

% As we mentioned in section 2, \cite{panda2021principal} changes the computation order from $\Sigma\mathbf{u} = X^TX\mathbf{u}$ to $X^T(X\mathbf{u})$, where $X$ represents the centered dataset. This approach eliminates the need to compute $\Sigma$ and avoids introducing matrix operations. However, the number of times it traverses all the samples is linearly proportional to the product of the target number of largest eigenvectors $k$ and the iterations of the Power Method $t$, with a multiplication complexity roughly estimated as $O(2tkd|X|)$. Furthermore, their computation depth required for the Power Method iteration is $t(\log(p)+3l+2)$, where $p$ is the smallest factor greater than or equal to the number of features $d$ and is also the maximum number of ciphertext slots for message encryption $N/2$, and $l$ is the total number of iterations required for the iterative algorithm of Inverse SquareRoot function. In our scheme, we first compute the covariance matrix, separating the iteration complexity of computing the Power Method from $|X|$. Additionally, our total multiplication depth is only linearly related to the iteration rounds and the target number of eigenvectors $k$, with a smaller constant factor and independent of the number of features.

\subsection{Iterative algorithm of inverse square root function} \label{sec: InvSRT}

The inverse square-root (InvSRT) function is used to perform vector normalization in the PowerMethod algorithm. One feasible approach to evaluate InvSRT homomorphically is to use iterative algorithms such as Newton's Method (see Algorithm \ref{alg:Newton}). These methods take the point $x$ for evaluation and an appropriate initial approximation of $\frac{1}{\sqrt{x}}$ as input and output a more accurate result after a finite number of iterations. We will use the term iterative InvSRT algorithm to refer to these methods in subsequent discussions. Providing a closer initial approximation for the iterative InvSRT algorithm can result in better convergence speed. Prior researches propose employing piecewise functions, Taylor expansions, and rational polynomials as functions to compute the initial approximation of the iterative InvSRT \cite{panda2022polynomial},\cite{qu4258571improvements}.

\begin{algorithm}
    \renewcommand{\algorithmicrequire}{\textbf{Input:}}
    \renewcommand{\algorithmicensure}{\textbf{Output:}}
    \caption{Newton's Method}
    \label{alg:Newton}
    \begin{algorithmic}[1]
    \REQUIRE
        $x_0$: target evaluation point of inverse square-root function;
        $y_0$: initial approximation of $\frac{1}{\sqrt{x_0}}$ 
    \ENSURE
        $y_d$: more precise approximation of $\frac{1}{\sqrt{x_0}}$
    \FOR{$i=1$ to $d$}
        \STATE $y_i\leftarrow \frac{1}{2}y_{i-1}(3-x_0y_{i-1}^2)$
    \ENDFOR
    \STATE return $y_d$         
    \end{algorithmic}
\end{algorithm}

\begin{comment}
\begin{theorem} 
\label{thm:GNequiv}
For any number of iterations $k$, algorithm \ref{alg:Goldschmidt} and \ref{alg:Newton} yield the same result.
\end{theorem}
\begin{proof} 
Assuming that for any iteration round less than $k$, we have $h_{i} = 1/2y_{i}$ and $g_{i} = x_{0}y_{i}$ for the algorithms mentioned above. For iteration round $i=1$, the above result holds. For iteration round $i = k$, we have $h_k = h_{k-1} + h_{k-1}(0.5-g_{k-1}h_{k-1}) = 1/2y_{k-1} + 1/2y_{k-1}(0.5 - x_0y_{k-1}\cdot 1/2y_{k-1})=1/2 y_k$, and $g_k = g_{k-1}+g_{k-1}(0.5-g_{k-1}h_{k-1})= x_{0}y_{k-1}+x_{0}y_{k-1}(0.5-1/2x_{0}y_{k-1}^2) = x_0y_{k}$.
\end{proof}    
\end{comment}

\section{Optimized Homomorphic Matrix Multiplication} \label{sec: Improved HMM}

In this section, we improve the homomorphic matrix multiplication proposed by Jinag \textit{et al.} to better serve our covariance matrix computation. Specifically, we first optimize the matrix multiplication with hoisting techniques to reduce its computation complexity (Section \ref{sec: Rotation compplexity}). Then, we analyze the space complexity of the hoisting-optimized matrix multiplication in parallel computation scenarios, during which we observe a problem that the large number of rotation keys required may occupy the memory and limit the performance of parallel computation (Section \ref{sec: parallel}). To address this problem, we introduce techniques to reduce these rotation keys (Section \ref{sec: RtkReduce}).

%We notice that hoisting techniques require a larger space in parallel computation scenarios.  We then introduce techniques to reduce the rotation keys required by the hoisting-optimized matrix multiplication since they may occupy the memory and limit the performance of parallel computation.

% and analyze the space complexity in parallel computation scenarios. We then observe a problem that the large number of rotation keys required by the matrix multiplication with hoisting techniques may occupy the memory and limit the parallel computation performance. To address this problem, we introduce techniques to reduce the rotation keys required by the matrix multiplication in Subsection \ref{sec: RtkReduce}.
%In this section, we improve the homomorphic matrix multiplication proposed by Jinag \textit{et al.} to better serve our covariance matrix computations. Specifically, we first optimize the matrix multiplication with hoisting techniques to reduce its computation complexity (in terms of the number of ciphertext rotations) and space complexity in parallel computation scenarios. We observe a problem that the large rotation key space required by the optimized matrix multiplication may occupy the memory and cause a memory saturation that limits the parallel computation performance. To address this problem, we introduce techniques for reducing the rotation key space for matrix multiplication in Subsection \ref{sec: RtkReduce}.

\subsection{Homomorphic Matrix Multiplication Equipped with Hoisting Techniques} \label{sec: HMMSec1}
%Therefore, in this subsection, we will incorporate these optimization techniques. We will first analyze the rotation complexity and rotation key space requirements of the hoisted matrix multiplication, and then discuss another crucial factor that affects space complexity, the number of intermediate ciphertexts in the context of matrix multiplication Parallelizability. %我们首先会分析 hoisting 技术加持的矩阵乘法的旋转复杂度和旋转密钥空间大小，另一重要的影响空间复杂度的因素：中间值密文数量将在接下来的矩阵乘法并行性中讨论与分析。

% We will first analyze the optimization principles of BSGS (Baby Step Giant Step) in the linear transformations used in Jiang \textit{et al.} (2018). Subsequently, we will evaluate the efficiency improvement and space requirements achieved by applying hoisting and double-hoisting to this matrix multiplication.

\subsubsection{Applying Hoisting Techniques to Linear Transformation Components} \label{sec: Rotation compplexity}

The hoisting techniques can be applied to the following linear transformations involved in the $n\times n$ matrix multiplication: $Z, T, \{C^k, R^k | 1\leq k < n\}$ (recall Equation \ref{eq4}). We initially focus on $Z$ and $T$. $Z$ comprises $2n-1$ non-zero diagonal vectors with indices ranging from $-n+1$ to $n-1$. On the other hand, $T$ consists of only $n$ non-zero diagonal vectors, with indices $k\cdot n$ for $0\leq k < n$. Both $Z$ and $T$ exhibit a well-behaved arithmetic progression in the indices of their non-zero diagonal vectors. This property enables the BSGS algorithm to reduce the number of rotations from $O(n)$ to $O(n_1+n_2)$, where $n_1,n_2$ denotes the inner and outer loop count. The double-hoisting technique can be naturally applied to the BSGS version of $Z$ and $T$, further reducing their rotation complexity.

The transformation $C^k$ has only two non-zero diagonal vectors: $k$ and $k-n$, while $R^k$ has only one non-zero diagonal vector: $n \cdot k$, for $1\leq k<n$. There is no need to apply BSGS optimization on these transformations due to the scarcity of their non-zero diagonals. Nonetheless, we can still enhance their efficiency using the hoisting technique. We observe that the rotation steps in the inner loop of $Z$ form a subsequence of the non-zero diagonal indices in the transformation set $\{C^k|0\leq k <n\}$. Similarly, rotation steps used for the inner loop of $T$ form a subsequence of the non-zero diagonal indices in the set $\{R^k|0\leq k <n\}$. This observation implies that the rotation keys of the double-hoisting $Z$ and $T$ are basically sufficient to perform the transformations $\{C^k, R^k|0\leq k <n\}$ with the hoisting technique (column shifting transformations require an extra rotation step $-n$ to perform rotations of steps $k-n,k=1,2,\dots,n-1$). In particular, we can achieve further optimization by emulating the strategy used in the double-hoisting technique for the transformations $C^k$. Deferring the final \textbf{ModDown} operations in the rotations associated with the two non-zero diagonal vectors, we multiply the two hoisted ciphertexts with the corresponding plaintext diagonal vectors in $R_{Q_lP}$ and aggregate the results. \textbf{ModDown} is then applied to the aggregated value. % 我们观察到对应旋转步数为 $1,2,\dots,n_{Z1}-1$ 的旋转密钥需要被提供给 double-hoisting 版本的 $Z$，他们正好构成了变换集 $\{C^k|0\leq k <n\}$ 的所有非零对角线序号的一个子序列。

We compared the ciphertext rotation complexity of various linear transformations in homomorphic matrix multiplication before and after applying the aforementioned hoisting technique (see Table \ref{tb:BSGS&dhBSGS}). %我们在表 1 中比较了同态矩阵乘法中各种线性变换在采取上述 hoisting 技术优化前后的计算复杂度（in terms of the ciphertext rotation complexity），
We can observe that the number of operations \textbf{Decompose} and \textbf{ModDown} is significantly decreased by the hoisting techniques. However, the minimal number of rotation keys required is increased by the inner loop counts of $Z$ and $T$. The increment of the rotation keys may cause memory limitation in a scenario where multiple matrix multiplication instances are desired to run in parallel. We will discuss this problem in the next subsection.

\subsubsection{Parallelizability and Space Complexity Analysis} \label{sec: parallel}
The matrix multiplication proposed by Jiang \textit{et al.} exhibits parallelizability. Considering Equation \ref{eq4}, we can perform transformation sets $\{C^kZ\mathbf{a}|0\leq k <n\}$ and $\{R^kT\mathbf{b}|0\leq k <n\}$ in parallel with two threads. This idea can be extended to a \textit{one-to-many} scenario where matrix $X$ is multiplied with multiple matrices $\{Y_i|i=1,2,\dots,d\}$. We can precompute and store $\{C^kZ\mathbf{x}|k=0,\dots,n\}$ (or $\{R^kT\mathbf{x}|k=0,\dots,n\}$), and perform the remaining steps of Equation \ref{eq4} in parallel for each $Y_i$ with $d$ threads. This approach saves the time required for $d-1$ computations of $\{C^kZ\mathbf{x}\}$ (or $\{R^kT\mathbf{x}\}$). We can further consider multiple parallel one-to-many computations, which can save more time when ${Y_i|i=1,2,\dots,d}$ also need to be repeatedly involved in calculations with a set of ${X_i|i=1,2,\dots,d}$. %We can further consider multiple parallel one-to-many computations, which 可以节省更多的时间当 $\{Y_i|i=1,2,\dots,d\}$ 也需要重复参与到与某组 $\{X_i|i=1,2,\dots,d\}$ 的计算中。 %考虑公式 \ref{eq4}，我们可以并行地计算 $\{C^kZ\mathbf{a}\}$ and $\{R^kT\mathbf{b}\}$. 我们可以把这种并行计算拓展到 “一对多” 的矩阵乘法场景中去：考虑 matrix $X$ is multiplied with multiple different matrices $\{Y_i|i=1,2,\dots,d\}$，那么我们可以 Precompute and store $\{C^kZ\mathbf{x}|k=0,\dots,n\}$ (or $\{R^kT\mathbf{x}|k=0,\dots,n\}$), and then compute the remaining steps of equation \ref{eq4} in parallel for each $Y_i$. 这样能节省 $d-1$ 次计算变换 $\{C^kZ\mathbf{x}\}$（或 $\{R^kT\mathbf{x}\}$）的时间。

However, we concern that sufficient space for the aforementioned parallel computation of matrix multiplication cannot be easily obtained. The space complexity of matrix multiplication is determined by the rotation key space of $Z$ and $T$, as well as the intermediate ciphertexts cached by all parallel instances of transformations. On one hand, the total number of rotation keys required in matrix multiplication is $n_{Z1} + 2n_{Z2} + n_{T1} + n_{T2} - 3$. On the other hand, the number of intermediate ciphertexts that need to be cached increases with the number of parallel instances of matrix multiplication. For example, $d$ instances of matrix multiplication collectively require caching $n$ regular ciphertexts and $d\cdot n_{T1}$ (or $d\cdot n_{Z1}$) hoisted ciphertexts under the one-to-many scenario. The scalability of increasing parallel instances might be hindered by the large space requirement for rotation keys.

\begin{table}
\renewcommand{\arraystretch}{1.5}
\begin{center}
\resizebox{\textwidth}{!} {
\begin{tabular}{@{}lll@{\hspace{0.2cm}}l@{}} \toprule

\text{Scheme} &  \text{LinTrans} &  \text{MinRotKeys} & \text{Complexity(Rotations)}\\ \midrule
$\text{Origin}$ & $Z$ & $2n_{2}+1$ &  $(n_{1}+2n_{2})\cdot(MS+Pm+Dp+MD)$\\ 
$\text{Optimizied}$ & $Z$ & $n_{1}+2n_{2}-2$ & \makecell[l]{ $(n_{1}+2n_{2})\cdot(MS+Pm)+$ $(2n_{2}+1)\cdot(Dp+MD)$} \\ 
$\text{Origin}$ & $T$ & $n_{2}+1$ &  $(n_{1}+n_{2})\cdot(MS+Pm+Dp+MD)$\\ 
$\text{Optimizied}$ & $T$ & $n_{1}+n_{2}-2$ &  \makecell[l]{$(n_{1}+n_{2})\cdot(MS+Pm)+$ $(n_{2}+1)\cdot(Dp+MD)$}\\ 
$\text{Origin}$ & $\{C^k|0\leq k <n\}$ & $2$ &  $(2n-1)\cdot (MS+Pm+Dp+MD)$\\ 
$\text{Optimized}$ & $\{C^k|0\leq k <n\}$ & $n_{Z1}+1$ & \makecell[l]{$(2n-1)\cdot (MS+Pm) + $$ (2n/n_{Z1}+2)\cdot Dp + (2n/n_{Z1}+n)MD$} \\ 
$\text{Origin}$ & $\{R^k|0\leq k <n\}$ & $1$ &  $(n-1)\cdot (MS+Pm+Dp+MD)$\\ 
$\text{Optimized}$ & $\{R^k|0\leq k <n\}$ & $n_{T1}$ & \makecell[l]{$(n-1)\cdot (MS+Pm)$ $+ (n/n_{T1}+1)\cdot Dp+ (n-1)\cdot MD $} \\ 
\bottomrule
\end{tabular}
}
\caption{Rotation Complexity and Minimal Rotation Keys requirement Comparison between BSGS and dh-BSGS scheme for Z, T linTrans. $n_1$ and $n_2$ denote the inner loop count and the outer loop count in the BSGS algorithm. The complexity of the linear transformations is indicated by the ciphertext rotations performed. More specifically, these rotations are decomposed into their four internal operations: \textbf{MultSum}$(MS)$, \textbf{Permute}$(Pm)$, \textbf{Decompose}$(Dp)$, and \textbf{ModDown}$(MD)$.} \label{tb:BSGS&dhBSGS}
\end{center}
\end{table}

\subsection{Rotation Key Reduction for Homomorphic Matrix Multiplication} \label{sec: RtkReduce}
We aim to minimize the number of rotation keys for homomorphic matrix multiplication, thereby freeing up space for the parallel computation of more instances. We introduce two types of rotation key reduction techniques in the following subsection. The first technique is a simple key substitution method. For an arithmetic sequence of rotations applied on a ciphertext, the simple key substitution constructs this sequence of rotation by repeatedly performing rotation with one of its subsequences, resulting in a reduction in the number of required rotation keys. However, this approach compromises the benefits brought by the hoisting technique, leading to a decrease in the overall speed of matrix multiplication. The second technique involves applying a decomposition called diagonal convergence decomposition to the transformations $Z$ and $T$. The number of non-zero diagonal vectors in the decomposed transformations is significantly reduced, resulting in a notable reduction in the number of keys while preserving the integrity of hoisting.

\subsubsection{Simple Key Substitution Method}

Recall that we only utilize the rotation keys from the inner loop of transformation $T$ to perform all ciphertext rotations of the row transformations (see Section \ref{sec: Rotation compplexity}). This approach is an instance of the simple key substitution method which can be described as follows: when applying a sequence of rotations to a ciphertext, if the rotation steps form an arithmetic progression starting from 0, then any subsequence starting from 0 of that arithmetic progression can be used to perform the rotations. 
However, using shorter subsequence results in fewer benefits from the hoisting technique, as the hoisting technique cannot be fully deployed across the entire sequence but only within individual subsequences. Employing this method to reduce the rotation keys required for the inner loop of $Z$ (or $T$) from $n_1$ to $n_1'$ leads to an additional cost of $n_1/n_1'$ \textbf{Decompose} and \textbf{ModDown} operations and a subsequent effect on the row (or column) transformations. Moreover, the simple key substitution method has marginal effectiveness when the inner loop count of the linear transformation is relatively small. This is because fewer iterations in the inner loop imply an increase in the number of outer loop iterations and the corresponding required keys.

\subsubsection{Decomposition of Diagonal Vector Convergence for Specific Permutation Matrix}
We present the second method for reducing the required rotation keys in matrix multiplication to overcome the limitations of the simple key substitution. Recall that the number of rotation keys required to perform a homomorphic linear transformation is essentially determined by the number of non-zero diagonal vectors in that linear transformation. This implies that if we can decompose $Z$ and $T$ into transformation with fewer non-zero diagonal vectors, we can reduce the number of rotation keys. Henceforth, we first establish a concept called diagonal convergence decomposition in definition \ref{def:2} to describe a decomposition method that reduces the number of non-zero diagonal vectors in linear transformations. This decomposition represents the original matrix as a sum and product of several matrices, such that the combined set of these decomposed matrices has fewer distinct non-zero diagonal indices compared to the original matrix. To better focus on the diagonals, we also define the diagonal coordinate of a matrix unit and the rules converting it back to the normal coordinate that locates units using the row and column index in Definitions \ref{def:1},\ref{def:3}.   %Thereby, 我们需要首先建立起一个概念来描述一种可以减少线性变换中的非零对角线数量的分解方法。在定义\ref{def:2}，我们定义了一种称为对角线收敛分解的方阵分解方法，它将原方阵分解成若干方阵的和与积的形式，分解后的这些方阵一共拥有的不同的非零对角线序号少于原方阵拥有的序号。由于我们在设计分解的时候主要关注非零对角线的数量，因此除了用行与列定位一个矩阵单位，我们还需要用对角线和行来定位一个矩阵单位，我们在定义 \ref{def:2,def:3} 中系统的给出了这两种定位的定义以及他们之间相互转换的方法。在我们的早期研究中我们发现，要构建一个对所有方阵都适用的对角线收敛分解方法并不是一件简单的事情，因此我们接下来只关注对 $Z$ 和 $T$ 设计对角线收敛分解。

Constructing a diagonal convergence decomposition method that applies to all matrices might not be a straightforward task. Therefore, we will specifically focus on designing the diagonal convergence decomposition for $Z$ and $T$. Both $Z$ and $T$ are permutation matrices, which are Boolean matrices with the properties that each row and each column contain exactly one 1 value. We observe that when decomposing a Boolean matrix $A$ of size $n\times n$ into a product of two Boolean matrices $A_2A_1$, for any non-zero diagonal $k$ in $A$ and its $i$-th element $A_{D(k,i)}$, if $A_{D(k,i)}$ is 1, there always exist a diagonal $k_1$ in $A_1$ and a diagonal $k_2$ in $A_2$ such that the elements in the $k_1$-th diagonal of $A_1$ are 1 in the same column as $A_{D(k,i)}$, the elements in the $k_2$-th diagonal of $A_2$ are 1 in the same row as $A_{D(k,i)}$, and $k = k_1+k_2 \mod{n}$ (see Theorem \ref{thm:2}). In Definition \ref{def:4}, we utilize Theorem \ref{thm:2} to construct a method for decomposing a single diagonal in $A$ into diagonals in $A_1$ and $A_2$.  %$Z$ 和 $T$ 都是置换矩阵，这种矩阵首先是布尔方阵，其次它们的每一行每一列都有且仅有一个 1 值。我们发现，在将一个 $n\times n$ 布尔矩阵 $A$ 分解成两个布尔矩阵乘积 $A_2A_1$ 形式的时候，对于 $A$ 中任意非零对角线 $k$ 的第 $i$ 个元素 $A_{D(k,i)}$，如果它是 1 值，那么总是存在 $A_1$ 上的对角线 $k_1$ 和 $A_2$ 上的对角线 $k_2$，使得 $k_1$ 对角线中和 $A_{D(k,i)}$ 同列的元素为 1，以及 $k_2$ 对角线中和 $A_{D(k,i)}$ 同行的元素为 1，且 $k = k_1+k_2 \mod{n}$（见定理 \ref{thm:2}）。继续考虑分解形式 $A=A_2A_1$，我们利用定理 \ref{thm:2} 在定义 4 中构造了一种针对将 $A$ 中单条对角线分解到 $A_1$ 和 $A_2$ 中去的方法。接下来我们将按照该分解方法对 $Z$ 和 $T$ 分别设计对角线收敛分解方法。 
\begin{definition}\label{def:2}
We refer to a matrix decomposition as a \textbf{diagonal convergence decomposition} when it expresses a square matrix as a product and sum of a series of matrices, and the total number of distinct diagonal vector indices in these matrices is smaller than the number of distinct diagonal vector indices in the original matrix.
\end{definition}

\begin{definition}\label{def:1}
For any $n\times n$ square matrix, we represent its column indices, row indices, and diagonal vector indices using the additive group $(\mathbb{Z}/n\mathbb{Z},+)$, where the indices of the diagonal vectors are defined based on the column indices of the elements in the first row of the matrix.
\end{definition}

\begin{definition}\label{def:3}
For an $n\times n$ square matrix, the element located at row $i$ and column $j$ possesses a \textit{normal coordinate} $N(i,j)\in (\mathbb{Z}/n\mathbb{Z})^2$, and a \textit{diagonal coordinate} $D(k,l)\in (\mathbb{Z}/n\mathbb{Z})^2$. Thus, there exists a mapping $f:(\mathbb{Z}/n\mathbb{Z})^2 \rightarrow (\mathbb{Z}/n\mathbb{Z})^2$ that transforms a normal coordinate to a diagonal coordinate:
    \begin{align}
    \begin{split}
        f(N(i,j))=D(j-i,i),\\
        f^{-1}(D(k,l))=N(l,k+l).
    \end{split}
    \end{align}
We denote the symbol $A_{N(i,j)}$ to represent the element of $A$ located at the normal coordinate $N(i,j)$, and $A_{D(k,l)}$ to represent the element of $A$ located at the diagonal coordinate $D(k,l)$.
\end{definition}

\begin{theorem}\label{thm:2}
Let $A$ be a boolean matrix that can be decomposed into $A = A_2A_1$, where $A_1$ and $A_2$ are both boolean matrices. For any 1 element $A_{N(i,j)}$ in $A$, there exists a unique index $0\leq a<n$ in $A_1$ and $A_2$ such that $A_{1N(a,j)}$ and $A_{2N(i,a)}$ are both 1. Furthermore, considering $D(k,i)=f(N(i,j))$, $D(k_1,a)=f(N(a,j))$, and $D(k_2,i)=f(N(i,a))$, we have $k_1+k_2 = k$.
\end{theorem}
\begin{proof}
    Given $A=A_2A_1$, we have $A_{N(i,j)} = \sum_{a=0}^{n-1} A_{1N(a,j)}\cdot A_{2N(i,a)}$. Since $A_1$ and $A_2$ are Boolean matrices, the equation $A_{N(i,j)}=1$ holds only when there exists a unique pair $(A_{1N(a,j)},A_{2N(i,a)})$ that equals $(1,1)$. According to Definition \ref{def:3}, we have $k_1+k_2 = (j-a)+(a-i) = k$.
\end{proof}

\begin{definition}\label{def:4}
For a boolean matrix $A$, $A_1$, and $A_2$ of size $n\times n$, decomposing the $k$th diagonal of $A$ onto the $k_1$th diagonal of the right matrix $A_1$, and the $k_2$th diagonal of the left matrix $A_2$ refers to the following process: For each 1 element $A_{D(k,l)}$ in the $k$th diagonal, set $A_{1D(k_1,k+l-k_1\mod{n})}$ to 1, which corresponds to the element on the $k_1$th diagonal with the same column index as $k$. Also, set $A_{2D(k_2,l)}$ to 1, which corresponds to the element on the $k_2$th diagonal with the same row index as $k$. Here, $k_1 + k_2 = k$.
\end{definition}

Before we proceed with the diagonal convergence decomposition of matrices $Z$ and $T$, let us establish some common ground. Let us assume that the discussions henceforth are based on the $n\times n$ matrix multiplication scenario, meaning matrices $Z$ and $T$ are always of size $n^2\times n^2$. Moreover, we define that if the maximum non-zero diagonal index of a matrix is $n'\leq 0$, it is equivalent to stating that all non-zero diagonal vectors of the resulting matrices belong to the set $[-n',n']\cap \mathbb{Z}$. \\

%\noindent\textbf{Decomposition of Linear Transformation $Z$} 
\paragraph{Decomposition of Linear Transformation $Z$}
The value of the $\ell$-th component in the $k$-th non-zero diagonal vector $\mathbf{z}_k[\ell]$ can be expressed as (all operations below involving non-zero diagonal vector indices are assumed to be in the residual class $(-n^2/2,n^2/2]$ of $(\mathbb{Z}/n^2\mathbb{Z},+)$):
\begin{align}\label{eq5}
    \begin{split}
          \mathbf{z}_k[\ell]=\begin{cases}1, &\text{if $k\geq 0$ }\text{ and $0\leq \ell-n\cdot k<(n-k)$}; \\ 1, & \text{if $k<0$ and $-k\leq \ell-(n+k)\cdot n<n;$} \\ 0,& \text{otherwise}.\end{cases}  
    \end{split}
\end{align}
We observe that the 1 values of $Z$ are symmetrically distributed on both sides of the 0th diagonal vector. They can be divided into $n$ submatrices of size $n\times n$, each having the 0th diagonal vector of $Z$ as its own 0th diagonal vector. These submatrices do not overlap with each other. Each submatrix contains $n$ elements with a value of 1. If we denote $S_i$, $0\leq i <n$, as the $n$ submatrices of diagonals of $Z$, then the 1 values in $S_i$ always fill the diagonals $i$ and $i-n$. Considering these characteristic, we design the diagonal convergence decomposition of $Z$ as follow: % have the following characteristic: $|[k_1]{n^2}|\leq \lceil\frac{n-1}{2}\rceil$ and $|[k_2]{n^2}|>\lceil\frac{n-1}{2}\rceil$.
\begin{enumerate}
    \item Set the expected maximum non-zero diagonal index as $n'$. Divide $Z$ into the format $Z = Z_1+Z_2$, where $Z_1$ is the matrix containing all non-zero diagonal vectors of $Z$ with an index no smaller than 0 and $Z_2$ contains the rest of the non-zero diagonal vectors. This step ensures that each submatrix of $Z_1$ and $Z_2$ has at most one non-zero diagonal vector, and each unit with value 1 belongs to a submatrix.
    \item For both $Z_1$ and $Z_2$, perform the following decomposition (assuming $Z_1$ or $Z_2$ is denoted as $Z'$):
        \begin{enumerate}
            \item For each non-zero diagonal with index $k$ in $Z'$, If the absolute value of the index is greater than $n'$, decompose it onto the $k-\text{sign}(k)\cdot n'$ diagonal of the left matrix $Z'_l$ and the $\text{sign}(k)\cdot n'$ diagonal of the right matrix $Z'_r$. If the absolute value of the index is no bigger than $n'$, decompose it onto the $n'$ diagonal of the $Z'_l$ and the $0$ diagonal of the $Z'_r$. %For each non-zero diagonal with index $k$ in $Z'$, 如果它的序号绝对值大于 $n'$，那么将其分解到左矩阵 $Z'_l$ 的第 $k-\text{sign}(k)\cdot n'$ 对角线，以及 $Z'_r$ 右矩阵的第 $\text{sign}(k)\cdot n'$ 对角线上。如果它的序号绝对值小于等于 $n'$ ，那么将其分解到 $Z'_l$ 的第 $n'$ 对角线，以及$Z'_r$ 的第 $0$ 对角线上。
            \item After iterating over all non-zero diagonals, if $Z'_l$ still contains non-zero diagonals greater than $n'$, replace $Z'$ with $Z'_l$ and repeat the decomposition. If no such diagonals exist, return the product of $Z'_l$ and all the right matrices obtained from the decomposition as the final decomposition result of $Z'$. %如果 $Z'_l$ 中仍存在大于 $n'$ 的非零对角线，那么用它替代 $Z'$，重复上述分解操作，如果不存在，则返回当前的 $Z'_l$ 与所有分解得到的右矩阵的乘积作为 $Z'$ 的分解结果。
        \end{enumerate}
Finally, we are able to decompose both $Z_1$ and $Z_2$ into the product of $\lceil\frac{n-1}{n'}\rceil$ matrix, respectively, as follows:
        \begin{align}
        Z = \prod_{i=1}^{\lceil\frac{n-1}{n'}\rceil} Z_{1i}+\prod_{i=1}^{\lceil\frac{n-1}{n'}\rceil} Z_{2i}.
        \end{align}
\end{enumerate}
The decomposed expressions of both products have the same characteristic: only the leftmost matrix has non-zero diagonal vectors with indices in the interval $[-n'+1,n'-1]\cap \mathbb{Z}$, while the remaining matrices have only non-zero diagonal vectors with indices $-n',0$ (or $n',0$). Therefore, the dominant complexity lies in the two leftmost linear transformations of the two matrix products. These two linear transformations still possess the same arithmetic progression property as $Z$, thus making them suitable for benefiting from the double-hoisting BSGS algorithm.

% FIXME:----------------------------------------- not complete statement for $Z$ ----------------------------

%todo: figure lack%

\begin{algorithm}
    \renewcommand{\algorithmicrequire}{\textbf{Input:}}
    \renewcommand{\algorithmicensure}{\textbf{Output:}}
    \caption{Diagonal Convergence Decompose for LinTrans Z (DCDforZ)}
    \label{alg5}
    \begin{algorithmic}[1]
    \REQUIRE \; \\
        $n$: dimension parameter to create a $Z$ with size $n^2\times n^2$;   \\
        $n'$: the target maximum non-zero diagonal index.
    \ENSURE
        $\{Z_{1i},Z_{2i}|Z = \prod_{i=1}^{\lceil\frac{n-1}{n'}\rceil} Z_i\}$: decomposed matrices. 
    \STATE Assume all indices of diagonal vectors are in residual class $(-n^2/2,n^2/2]$.
    \STATE Set $Z_1$ as the matrix containing all non-zero diagonal vectors of $Z$ with an index no smaller than $0$, and set $Z_2$ as the matrix containing the rest of the non-zero diagonal vectors of $Z$. 
    \FOR{t=1,2}
            \FOR{$i=1$ to $\lceil\frac{n-1}{n'}\rceil-1$}
                \STATE Initialize matrices $Z_{t}^{(l)},Z_{t}^{(r)}$ as zero matrices.
                \FOR{each non-zero diagonal vector with index $|k| \geq n'$ in $Z_t$} 
                    \STATE Decompose the vector to the $n'$-th non-zero diagonal vector of the right matrix $Z_{t}^{(r)}$ and the $(k-\text{sign}(k)\cdot n')$-th diagonal of the left matrix $Z_{t}^{(l)}$.
                \ENDFOR
                \FOR{each non-zero diagonal vector $k<n'$ in $Z_t$}
                    \STATE Decompose the vector to the $0$-th non-zero diagonal vector of the right matrix $Z_{t}^{(r)}$ and the $k$-th diagonal of the left matrix $Z_{t}^{(l)}$.
                \ENDFOR
                \STATE (At this point, $Z_t = Z_{t}^{(l)}Z_{t}^{(r)}$)    
                \IF{$i\not = \lceil\frac{n-1}{n'}\rceil-1$}
                    \STATE $Z_{t(\lceil\frac{n-1}{n'}\rceil-i+1)} \leftarrow Z_{t}^{(r)}$ , $Z_t\leftarrow Z_{t}^{(l)}$ 
                \ELSE
                    \STATE  $Z_{t(\lceil\frac{n-1}{n'}\rceil-i+1)} \leftarrow Z_{t}^{(r)}$ , $Z_{t1}\leftarrow Z_{t}^{(l)}$ 
                \ENDIF
            \ENDFOR
    \ENDFOR
    \RETURN $\{Z_{1i},Z_{2i}|Z = \prod_{i=1}^{\lceil\frac{n-1}{n'}\rceil} Z_i\}$    
    \end{algorithmic}
\end{algorithm}

%\noindent \textbf{Decomposition of Linear Transformation $T$} 
\paragraph{Decomposition of Linear Transformation $T$}

Transformation $T$ consists of $n$ non-zero diagonal vectors, whose indices are $\{k\cdot n|0\leq k < n\}$. For the $k\cdot n$-th non-zero diagonal vector $\mathbf{t}_{kn}$ of $T$, the value of its $\ell$-th component $\mathbf{t}_{kn}[\ell]$ can be expressed as (all operation below involving diagonal vector indices are assumed to be in the residual class $[0,n^2-1]$ of $(\mathbb{Z}/n^2\mathbb{Z},+))$):
\begin{align}\label{eq6}
    \begin{split}
        \mathbf{t}_{kn}[\ell] = \begin{cases}1, &\text{if $\ell \in \{k+n\cdot i|0\leq i < n\}$ };  \\ 0,& \text{otherwise}.\end{cases}
    \end{split}
\end{align}
% FIXME: ------------------------------------- Incomplete statement for T---------------------------------
Although the units with a value of 1 in $T$ can not be divided into submatrices like $Z$, the non-zero diagonals of $T$ have a consistent spacing between the units with a value of 1. Specifically, we observe that the row indices of units with a value of 1 in the same non-zero diagonal vector are congruent modulo $n$. This implies that the row sequence of 1-valued units on the $kn$-th non-zero diagonal of $T$ forms a coset $k+\left<n\right>$ of an $n$-order subgroup $\left<n\right>$ in $\mathbb{Z}/n^2\mathbb{Z}$. According to Definition \ref{def:3}, their column sequences also form the same coset. If we attempt to decompose any diagonal $kn$ in $T$ onto a diagonal of index divisible by $n$ in the right matrix $T_1$ and the left matrix $T_2$, then the 1-valued units on these two diagonals in the left and right matrices will inherit the property of forming the coset $k+\left<n\right>$. Furthermore, if any two diagonals $kn$ and $k'n$ with the coset properties $i+\left<n\right>$ and $j+\left<n\right>$, respectively, are decomposed onto the $k_1n$ diagonal of the right matrix (or left matrix), then the 1-valued units in that diagonal will form a set $(i+\left<n\right>)\cup j+\left<n\right>)$. Based on these observations, we can repeatedly decompose the diagonals in $T$ that have indices greater than the expected maximum non-zero diagonal index $n \cdot n'$ into diagonals $k-n \cdot n'$ and $n \cdot n'$, ultimately decomposing $T$ into the product of ${\lceil \frac{n-1}{n'} \rceil}$ matrices:
\begin{align}
T = \prod_{i=1}^{\lceil \frac{n-1}{n'} \rceil} T_{i}.
\end{align}
Similar to the decomposition of Z, Only the leftmost matrix in the product contains non-zero diagonals with indices in $\{kn|k=0,1,\dots,n'\}$, while all other right matrices have only non-zero diagonals with indices $n'n$ and $0$. This implies that the complexity of the product is dominated by the complexity of the leftmost linear transformation, whose diagonal indices belong to a smaller range compared to the original approach. The detailed decomposition process is provided in Algorithm \ref{alg:DCD4T}.
%Although the units with a value of 1 in $T$ can not be divided into submatrices like $Z$, the non-zero diagonals of $T$ have a consistent spacing between the units with a value of 1. Specifically, we observe that the row indices of units with a value of 1 in the same non-zero diagonal vector are equivalent to each other 模 $n$。 这暗示了 $T$ 的第 $kn$ 非零对角线中所有 1 值单元的行序列构成了 $\mathbb{Z}/n^2\mathbb{Z}$ 中 $n$ 阶子群 $\left<n\right>$的一个陪集 $k+\left<n\right> = \{k+i\cdot n | i\in \mathbb{Z}/n^2\mathbb{Z}\}$. 根据定义 \ref{def:3}，可知他们的列序列也构成了相同的陪集。如果我们尝试将 $T$ 中任意对角线 $kn$ 分解到右矩阵 $T_1$ 和 左矩阵 $T_2$ 序号为 $n$ 的倍数的对角线上，那么左右矩阵上的这两条对角线上的所有 1 值单元也将继承构成陪集 $k+\left<n\right>$ 的性质。%Furthermore, 如果任意两个具有陪集性质 $i+\left<n\right>,j+\left<n\right>$ 的对角线 $kn$ 和 $k'n$ 都被分解到了右矩阵（或左矩阵）的 $k_1n$ 对角线上时，该对角线将具有陪集性质 $i+j+\left<n\right>$。 基于上述观察，我们可以像分解 $Z$ 那样，将 $T$ 中大于期望的最大非零对角线序号 $n\cdot n'$ 的那些对角线 $k$ 不断地分解 $k-n\cdot n'$ 和 $n\cdot n'$ 上，最终将 $T$ 分解成若干矩阵的乘积的形式。详细的分解过程在算法 \ref{alg:DCD4T} 中给出。
\begin{algorithm}
    \renewcommand{\algorithmicrequire}{\textbf{Input:}}
    \renewcommand{\algorithmicensure}{\textbf{Output:}}
    \caption{Diagonal Convergence Decompose for LinTrans T(DCDforT)}
    \label{alg:DCD4T}
    \begin{algorithmic}[1]
    \REQUIRE \; \\
        $n$: the dimension parameter for creating a $T$ with size $n^2\times n^2$;\\
        $n'\cdot n$: the target maximum non-zero diagonal index.
    \ENSURE 
        $\{T_i| T= \prod_{i=1}^{\lceil \frac{n-1}{n'} \rceil} T_i\}$ decomposed matrices of linear transformation $T$. 
    \STATE Assume that all indices of diagonal vectors are in residual class $[0,n^2-1]$
    \FOR{$i=1$ to $\lceil\frac{n-1}{n'}\rceil-1$}
        \STATE Initialize $T^{(0)}$ and $T^{(1)}$ as zero matrices.
        \STATE Traverse the $i\cdot n$-th non-zero diagonal of $T$, where $0\leq i \leq n'$, and decompose it onto the 0th diagonal of the right matrix $T^{(r)}$ and the $i\cdot n$-th diagonal of the left matrix $T^{(l)}$.
        \STATE Traverse the $i\cdot n$-th non-zero diagonal of $T$, where $n' < i \leq (n-1)$, and decompose it onto the $n'\cdot n$-th diagonal of the right matrix $T^{(r)}$ and the $(i-n')\cdot n$-th diagonal of the left matrix $T^{(l)}$.
        \STATE (At this point,$T = T^{(l)}T^{(r)}$)
        \IF{$i\not = \lceil\frac{n-1}{n'}\rceil-1$}
            \STATE $T_{\lceil \frac{n-1}{n'} \rceil-i+1} \leftarrow T^{(r)}$, $T\leftarrow T^{(l)}$ 
        \ELSE
            \STATE  $T_{\lceil \frac{n-1}{n'} \rceil-i+1} \leftarrow T^{(r)}$, $T_{1}\leftarrow T^{(l)}$
        \ENDIF
    \ENDFOR
    \STATE return $\{T_i| T= \prod_{i=1}^{\lceil \frac{n-1}{n'} \rceil} T_i\}$   
    \end{algorithmic}
\end{algorithm}

After formulating the diagonal convergence decomposition for $Z$ and $T$, we proceed to analyze its impact on the number of rotation keys and rotation complexity. Let $n_1$ and $n_2$ continue to represent the inner loop and outer loop counts of the original double-hoisting BSGS applied to $Z$ (or $T$). Since both the diagonal convergence decompositions for $Z$ and $T$ are able to reduce the non-zero diagonal vectors from $n-1$ to $n'$, for a fixed $n_1$, the decomposition shrinks the number of keys required for the outer loop (of $Z$ and $T$) with a factor of $\frac{n_2'}{n_2-1}$, where $n_2' = n'/n_1$. As for the rotation complexity, the primary complexity for $Z$ and $T$ after decomposition arises from the leftmost matrices in the matrix products, which require $(n_1+n_2')(MS+Pm) + (n_2'+1)(Dp+MD)$. The extra $(n/n'-1)\cdot(MS+Pm+Dp+MD)$ is attributed to the $n/n'-1$ matrices on the right side. We compare the number of rotation keys and computational complexity (in terms of internal rotation operations) among the diagonal convergence decomposition, simple key substitution method, and the dh-BSGS algorithm without any key simplification strategy for the $Z$ and $T$ transformations, as shown in Table \ref{tb:ZTRotCplx}. More detailed experimental data obtained from specific configurations will be discussed in the implementation section. Additionally, we present the complete process of ciphertext decomposition matrix multiplication with diagonal convergence decomposition in Algorithm \ref{alg7}, omitting internal optimization details such as double-hoisting and hoisting, and providing a high-level illustrative flow. %After completing the design of diagonal convergence decomposition for $Z$ and $T$, we then analyze its impact on the number of rotation keys and the rotation complexity. 设 $n_1$ 和 $n_2$ 是应用在 $Z$ （或 $T$）上的原始 BSGS 算法的内层循环和外层循环数。Since both the diagonal convergence decompositions for $Z$ and $T$ can shrink the non-zero diagonal vectors from $n-1$ to $n'$， thus for a fixed inner loop count $n_1$, the decomposition 将 $Z$ 和 $T$ 中外层循环所需的密钥数量缩小至原来的 $\frac{n_2'}{n_2-1}$ 其中 $n_2' = n'/n_1$。 至于旋转复杂度，$Z$ 和 $T$ 在经过分解后，其主要的复杂度均来自于矩阵乘积式中的最左侧矩阵，它们都需要大致计算复杂度 $(n_1+n_2')(MS+Pm) + (n_2'+1)(Dp+MD)$，而剩下的 $(n/n'-1)\cdot(MS+Pm+Dp+MD)$ 则是由右侧的 $n/n'-1$ 个矩阵所贡献的。我们 compared the number of rotation keys and computational complexity (in terms of internal rotation operations) of the diagonal convergence decomposition, simple key substitution method, and the dh-BSGS algorithm without any key simplification strategy for the $Z$ and $T$ transformations in Table \ref{tb:ZTRotCplx}. More detailed experimental data produced by concrete settings will be discussed in the implementation section. We also provide the complete process of ciphertext decomposition matrix multiplication with diagonal convergence decomposition in Algorithm \ref{alg7}, omitting the internal optimization details such as double-hoisting and hoisting, and presenting a high-level illustrative flow.
%表中的 $n_1$ 和 $n_2$ 分别代表原始 BSGS 线性变换的内层和外层循环数。Simple key substitution 的复杂度我们已在上文分析过。由于 $Z$ 和 $T$ 在经过 diagonal convergence decomposition 后，其主要的复杂度均来自于矩阵乘积式中的最左侧矩阵，它们都需要大致计算复杂度 $(n_1+n_2')(MS+Pm) + (n_2'+1)(Dp+MD)$，而剩下的 $(n/n'-1)\cdot(MS+Pm+Dp+MD)$ 则是由右侧的 $n/n'-1$ 个矩阵所贡献的。可以看出，当 $n_1$ 相对 $n_2$ 较小时， diagonal convergence decomposition 对空间和计算复杂度都有更好的优化效果。detailed experimental data produced by concrete settings will be discussed in the implementation section. 

% 2* ( (n_1+n_2')(MS+Pm) + (n_2'+1)(Dp+MD) + n/(n_1*n_2')(MS+Pm+Dp+MD) )

\begin{table}
\renewcommand{\arraystretch}{1.5}
\resizebox{\textwidth}{!}{
\begin{tabular}{@{}lllll@{}} \toprule
\textbf{Scheme} &\textbf{LT}& \textbf{MaxDiagNo.} & \textbf{MinRotKeys} & \textbf{Complexity(Rotations)} \\
\midrule
$\text{dh-BSGS}$ & $Z$ & $n-1$ & $n_{1}+2n_{2}-2$ & \makecell[l]{$(n_{1}+2n_{2})\cdot(MS+Pm)+$ $(2n_{2}+1)\cdot(Dp+MD)$}\\

$\text{SmpKeySub}$ & $Z$ & $n-1$ & $n_1'+2n_2-2$ &  \makecell[l]{$(n_{1}+2n_{2})\cdot(MS+Pm)+$ $(2n_{2}+n_1/n_1'+1)\cdot(Dp+MD)$} \\ 

$\text{DCDmp}$ & $Z$ & $n'=n_1\cdot n_2'$  & $n_1+ 2n_2'-1$ & \makecell[l]{$2(n_1+n_2'+n/n'-1)\cdot (MS+Pm) + 2(n_2'+n/n')\cdot (Dp+MD)$}\\  

$\text{dh-BSGS}$ & $T$ & $(n-1)\cdot n$ & $n_1+n_2-2$ & \makecell[l]{$(n_1+n_2)\cdot(MS+Pm)+$$(n_2+1)\cdot (Dp+MD)$}    \\

$\text{SmpKeySub}$ & $T$ & $(n-1)\cdot n$ & $n_1'+n_2-2$ & \makecell[l]{$(n_1+n_2)\cdot (MS+Pm)+$$(n_2+n_1/n_1'+1)\cdot (Dp+MD)$} \\

$\text{DCDmp}$ & $T$ & $n'\cdot n=n_1\cdot n_2' \cdot n$ & $n_1+n_2'-1$ & \makecell[l]{$(n_1+n_2'+n/n'-1)\cdot (MS+Pm) + (n_2'+n/n')\cdot (Dp+MD)$}\\
\bottomrule
\end{tabular}
}
\caption{Rotation Complexity and Minimal Rtks Number Comparison of LinTrans $Z,T$ among original double-hoisting BSGS(dh-BSGS), dh-BSGS with Simple Key Substitution(SmpKeySub) and dh-BSGS with Diagonal Vector Convergence Decomposition(DCDmp). \textbf{MaxDiagNo.} denotes the maximum non-zero diagonal index in the lineartransformation(s), and \textbf{MinRotKeys} indicates the minimum number of keys required to perform the lineartransformation(s).} \label{tb:ZTRotCplx}
\end{table}

\begin{algorithm}[h]
    \renewcommand{\algorithmicrequire}{\textbf{Input:}}
    \renewcommand{\algorithmicensure}{\textbf{Output:}}
    \caption{Matrix Multiplication using Diagonal Convergence Decomposition (DMatrixMult)}
    \label{alg7}
    \begin{algorithmic}[1]
    \REQUIRE \; \\
        $\text{ct}(A),\text{ct}(B)$: ciphertexts encrypting matrices $A$ and $B$ with size $n\times n$; \\
        $n'_Z$: expected maximum non-zero diagonal index for $Z$; \\
        $n\cdot n_T'$: expected maximum non-zero diagonal index for $T$.
    \ENSURE ;\ \\
        $\text{ct}(AB)$: encrypted multiplication result between matrices $A$ and $B$ 
        \STATE $\{Z_{1i}\},\{Z_{2i}\}\leftarrow\text{DCDforZ}(n^2,n'_Z)$ 
        \STATE $\{T_i\}\leftarrow \text{DCDforT}(n^2,n\cdot n_T')$
        \STATE $ct(A_1),ct(A_2)\leftarrow ct(A)$
        \STATE $ct(B^{(0)})\leftarrow ct(B)$
        \FOR{$Z_{1i}$ in $\{Z_{1i}\}$}
            \STATE $\text{ct}(A_1)\leftarrow \text{LinTrans}(\text{ct}(A_1),Z_{1i})$ 
        \ENDFOR
        \FOR{$Z_{2i}$ in $\{Z_{2i}\}$}    
            \STATE $ct(A_2)\leftarrow \text{LinTrans}(ct(A_2),Z_{2i})$
        \ENDFOR
        \STATE $ct(A^{(0)})\leftarrow \text{Add}(ct(A_1),ct(A_2))$
        \FOR{ $T_{i}$ in $\{T_{i}\}$}
            \STATE $\text{ct}(B^{(0)})\leftarrow \text{LinTrans}(\text{ct}(B^{(0)}),T_{i})$ 
        \ENDFOR
        \STATE $\text{ct}(AB) \leftarrow \text{Mul}(\text{ct}(A^{(0)}),\text{ct}(B^{(0)}))$
        \FOR{ $1 \leq k < n$ }
            \STATE  $\text{ct}(A^{(k)}) \leftarrow \text{LinTrans}(\text{ct}(A^{(0)},C^k)$ 
            \STATE $\text{ct}(B^{(k)}) \leftarrow \text{LinTrans}(\text{ct}(B^{(0)},R^k)$ 
	    \STATE $\text{ct}(AB^{(k)})\leftarrow \text{Mul}(\text{ct}(A^{(k)}),\text{ct}(B^{(k)}))$ 
	    \STATE $\text{ct}(AB) \leftarrow \text{Add}(\text{ct}(AB^{(k)}),\text{ct}(AB))$
        \ENDFOR
    \STATE return $\text{ct}(AB)$         
    \end{algorithmic}
\end{algorithm}

%\noindent \textbf{Further Discussion}

%\noindent 我们观察到 \cite{jiang2018secure} 中提出实现矩阵转置的线性变换具有和 $T$ 类似的非零对角线分布，因此我们的对角线收敛分解策略或许可以应用在转置线性变换上，减少转置操作所需的密钥数量，并有效提升其速度，这一部分内容将在附录进行讨论。

\section{Homomorphic Covariance Matrix Computation for Large Datasets}
In this section, We move on to present the method to pack large datasets into ciphertexts which enables seamless execution of the homomorphic matrix multiplication (see Section \ref{sec: Ciphertext Packing}) and a homomorphic circuit that efficiently computes the covariance matrix of the packed and encrypted dataset.(see Section \ref{sec: Cov}.

%In the previous section, we presented the optimization for matrix multiplication. In this section, we will first introduce a method for ciphertext packing of large datasets, enabling seamless execution of matrix multiplication on encrypted data. Next, by leveraging homomorphic matrix multiplication and other homomorphic operations, we propose an algorithm that efficiently computes the covariance matrix of a dataset using only encrypted dataset elements.

\subsection{Ciphertext Packing}\label{sec: Ciphertext Packing}

Let $X$ be an $s \times t$ dataset, where $s$ represents the number of samples and $t$ represents the number of features. We partition $X$ into square submatrices with size $n\times n$, where $n^2$ equals the CKKS scheme parameter $N/2$. This allows us to encode each submatrix into one ciphertext. Features of the dataset are divided into $k$ partitions, where $k = \lceil \frac{t}{n} \rceil$. In this way, each ciphertext will contain one partition of features from $n$ samples/records. Note that $t$ is not always divisible by $n$, so zero-padding is necessary. Two possible methods for padding are as follows:
% Our strategy involves partitioning $X$ into square matrices, as this facilitates the execution of homomorphic square matrix operations. Therefore, we consider schemes where the number of message slots is $n^2 = N/2^i$, where $1 \leq i < \text{log}(N)$, allowing us to treat the ciphertext as an $n \times n$ matrix. We will divide the features into $k$ partitions, where $k = \lceil \frac{t}{n} \rceil$. Each ciphertext will contain one partition of features from $n$ samples/records, meaning that $k$ ciphertexts will collectively represent all the feature values of $n$ samples. Note that $t$ is not always divisible by $n$, so zero-padding is necessary. Two possible methods for padding are as follows:
\begin{enumerate}
    \item The first $k-1$ columns of submatrices each contain $n$ features, while the last column of submatrices contains the last $t-k\cdot n$ features and pads the rest of its space with zeros. This approach expands the feature count from $t$ to $k \cdot n$ and treats each ciphertext as an $n \times n$ matrix representing $n$ samples with $n$ features.
    \item All columns of submatrices store only $n-p$ valid features and pad the rest of their space with zeros, where $p=\lfloor(t - k\cdot n)/k \rfloor$. This approach allows each ciphertext to be viewed as an $n \times (n - p)$ matrix, taking advantage of some possible optimizations for rectangular matrix multiplication while still enabling operations on square matrices.
\end{enumerate}
In summary, $\lceil s/n \rceil \cdot k$ ciphertexts are required to pack the entire dataset $X$.

\subsection{Homomorphic Covariance Matrix Computation for Large Datasets}\label{sec: Cov}

We present a method for computing the covariance matrix of an encrypted dataset. Assume that we have partitioned $X$ into $\lceil s/n \rceil \cdot k$ ciphertexts, with each ciphertext representing an $n \times n$ submatrix of $X$. We denote $X_i$ as the $i$-th column submatrix of $X$, where $0 \leq i < k$. We also denote $X_i[\ell]$ as the $\ell$-th submatrix of $X_i$ and $\text{ct}(X_i[\ell]$ as the ciphertext of  $X_i[\ell]$. We need to calculate the mean 'vector' $\mu$ of the dataset if it is not centered. $\mu$ is represented by $k$ submatrices, where each column contains $n$ replication of one feature's mean value. $0\leq i<k$, denote $\mu[i]$ as the $i$-th submatrix of the mean vector. $\mu$ can be computed through the following equation: 
\begin{align}
  \mu[i] = \text{Aggregate}\left( \sum_{0\leq \ell < \lceil s/n \rceil} X_i[\ell];0 \right),  
\end{align}
where $\text{Aggregate}(X;\text{axis}), a \in {1,0}$ represents aggregating a matrix $X$ by columns ($\text{axis}=1$) or by rows ($\text{axis}=0$).This algorithm requires at most 1 multiplication depth (see Algorithm \ref{alg:agg}). 
\begin{algorithm}
    \renewcommand{\algorithmicrequire}{\textbf{Input:}}
    \renewcommand{\algorithmicensure}{\textbf{Output:}}
    \caption{Aggregate}
    \label{alg:agg}
    \begin{algorithmic}[1]
    \REQUIRE \; \\
        $ct(X)$: encrypted matrix $X$ with size $\text{row}\times \text{col}$; \\
        $\text{axis}$: aggregation axis.
    \ENSURE \; \\
        $ct(\text{Aggregate}(X;\text{axis}))$ 
    \STATE $ct(A)\leftarrow ct(X)$ 
    \IF{$\text{axis} == 0$}
	\FOR{$i\gets col; i< \text{row}\cdot \text{col};i= (i<<1)$}
		\STATE $ct(A)\leftarrow \text{Add}(\text{Rot}(ct(A);i),ct(A))$ 
        \ENDFOR
    \ENDIF
    \IF{$\text{axis} == 1$}
	\FOR{$i\gets 1;i<\text{col}; i = (i<<1)$}
		\STATE $ct(A)\leftarrow \text{Add}(\text{Rot}(ct(A);i),ct(A))$ 
        \ENDFOR
        \STATE $ct(A) \leftarrow \text{Mult}(ct(A), pt(M))$, where $M$ is a mask matrix with the same dimensions as $X$, except for the first column elements which are all 1 and other elements are 0
	\FOR{ $i:= -1;i>-col ; i=(i<<1)$}
		\STATE $ct(A)\leftarrow \text{Add}(\text{Rot}(ct(A);i),ct(A))$ 
        \ENDFOR
    \ENDIF
    \RETURN $\text{ct}(A)$         
    \end{algorithmic}
\end{algorithm}
\begin{comment}
\begin{theorem}
The multiplication depth of Algorithm \ref{alg:agg} is at most 1, and the time complexity is dominated by the number of rotations, which is $O(\text{log}(\text{row}))$ or $O(\text{log}(\text{col}))$.
\end{theorem}    
\end{comment}
Next, matrix $\mu^T \mu$ can be simply computed by transpositions and coordinate-wise multiplications:
\begin{align}
    \mu^T\mu_i[j] = \mu[j]^T \odot\mu[i].
\end{align}
We then move on to compute each submatrix unit in $X^TX$. For the submatrix $X^TX_i[j]$, we iterate through each submatrix in column $X_i$ and multiply it with the transposition of the corresponding submatrix in column $X_j$:
\begin{align}
  X^TX_i[j] = \sum_{0\leq \ell < \lceil s/n \rceil} X_j[\ell]^T X_i[\ell].  
\end{align}
Here, each matrix multiplication can be performed in ciphertext space using our improved homomorphic matrix multiplication. If $X$ is already centered, then $X^TX$ represents the covariance matrix. Otherwise, $X^TX$ should be subtracted by $\mu^T\mu$ to become the covariance matrix. It is worth noting that directly computing the covariance matrix homomorphically using the equation above is not the most efficient implementation. We provide the following optimizations: 
\begin{enumerate}
    \item Notice that for a given $0\leq j <k$, there is a one-to-many matrix multiplication scenario between $X_j[\ell]$ and matrix set $\{X_i[\ell]|0\leq i <k \}$ in the equation above. We can perform the inner product operations of $k$ columns of submatrices in parallel: $\sum_{0\leq \ell < \lceil s/n \rceil} X_j[\ell]^T X_i[\ell], 0\leq i <k$, and save time by storing $\{C^kZ \text{ct}(X_j[\ell]^T)|0\leq \ell <\lceil s/n \rceil\}$. 
    %Considering that in Section \ref{sec: Improved HMM}, we mentioned that when multiplying a matrix $A$ with a series of matrices $B_1, B_2, \dots$, we can reduce more than half of the matrix multiplication overhead by storing $\{C^kZ\mathbf{a}|k=0,\dots,n\}$. Therefore, for a given $j$, we can perform the inner product operations of the $k$ column submatrices in parallel: $\sum_{0\leq \ell < \lceil s/n \rceil} X_j[\ell]^T X_i[\ell], 0\leq i <k$, and save nearly half of the computation time by storing $\{C^kZ ct(X_j[\ell]^T)|0\leq \ell <\lceil s/n \rceil\}$. 
    \item We can further cache all ${R^kT ct(X_i[\ell])|0\leq \ell <\lceil s/n \rceil,0\leq j <k}$, so that each $X_i[\ell]$ only needs to undergo one $T$ transformation and $n-1$ row transformations throughout the entire covariance matrix computation process.
    %FIXME: add some points.
    \item Based on the fact that the covariance matrix is a symmetric matrix, it suffices to calculate the results of one side of the diagonal and transpose them to obtain the results for the other side, saving nearly half of the covariance computation time. 
\end{enumerate}
However, the ciphertexts required to be stored in Optimization 1 and 2 usually cannot be fully cached in memory due to space limitations. Therefore, it is often necessary to split the dataset into suitable segments by rows and apply Optimization 1 and 2 segment by segment.

We present the computation of the entire encrypted covariance matrix with the aforementioned Optimization 3 in Algorithm \ref{alg:Cov}. Optimization 1 and 2 are not explicitly shown in the algorithm since they involve restructuring the internal computation order of the DMtrixMult function. Theorem \ref{thm:CovCplx} provides the multiplication depth and computation complexity of Algorithm \ref{alg:Cov}, where the computation complexity is measured at the granularity of matrix multiplications since matrix multiplication can be implemented using different methods with varying complexities (such as whether to use Optimization 1 and 2).
\begin{algorithm}
    \renewcommand{\algorithmicrequire}{\textbf{Input:}}
    \renewcommand{\algorithmicensure}{\textbf{Output:}}
    \caption{Homomorphic Covariance Matrix Computation (HCovMtrx)}
    \label{alg:Cov}
    \begin{algorithmic}[1]
    \REQUIRE \; \\
        ${ \text{ct}(X_j[\ell]) | 0\leq j < k, 0\leq \ell < \lceil s/n \rceil }$: partitioned and encrypted dataset; \\
        $n_Z$: maximum diagonal index for transformation $Z$;\\
        $n_T$: maximum diagonal index $n_T$ for transformation $T$;\\
        % $G$: permutation matrix for matrix transposition with size $n^2\times n^2$.
    \ENSURE\; \\
        ${\text{ct}(\text{Cov}i[j])| 0\leq i,j < k}$: partitioned and encrypted covariance matrix.
    \FOR{ $0\leq j < k$ }
        \STATE Take $\text{ct}(X_j[\ell]),0\leq \ell < \lceil s/n \rceil$, $\text{ct}(\mu[j])\leftarrow \sum_{0\leq \ell < \lceil s/n \rceil}X_i[\ell]$ 
        \FOR{$n \leq t <n^2$ }
            \STATE $ct(\mu[j])\leftarrow \text{Add}(\text{ct}(\mu[j]),\text{Rot}(\text{ct}(\mu[j]);t))$ 
            \STATE $t\leftarrow t<<1$ 
        \ENDFOR
        \STATE $\text{ct}(X_j[\ell]^T)\leftarrow \text{LinTrans}(\text{ct}(X_j[\ell]),\text{G})$ \COMMENT {perform transposition with permutation matrix $G$ in parallel for $0\leq \ell < \lceil s/n \rceil$}
        \FOR{ $0\leq i <j$ } %lower triangular submatrix
           \STATE $\text{ct}(X_j[\ell]^T X_i[\ell])\leftarrow \text{DMtrxMul}(\text{ct}(X_j[\ell]^T),\text{ct}(X_i[\ell]),n,n_Z,n_T)$ \COMMENT {compute in parallel for $0\leq \ell < \lceil s/n \rceil$}
           \STATE $\text{ct}(X^TX_i[j])\leftarrow \sum_{0\leq i <k} \text{ct}(X_j[\ell]^TX_i[\ell])$; \COMMENT{aggregate the parallel computation results} 
        \ENDFOR
    \ENDFOR
    \FOR{ $0 \leq j < k$ }
        \FOR{$j \leq i <k$} % upper triangular submatrix
            \STATE $ct(X^TX_i[j])\leftarrow \text{LinTrans}(ct(X^TX_j[i]),\text{G})$ 
        \ENDFOR
    \ENDFOR
    \FOR{ $0\leq j <k$ }
        \STATE $\text{ct}(\mu[j]^T)\leftarrow \text{LinTrans}(\text{ct}(\mu[j]),T)$ 
        \FOR{ $0\leq i <k$}
            \STATE $\text{ct}(\mu^T\mu_i[j])\leftarrow \text{Mul}(\text{ct}(\mu[j]^T),\text{ct}(\mu[i]))$
            \STATE $\text{ct}(\text{Cov}_i[j]) \leftarrow \text{Add}(ct(X^TX_i[j]),ct(-\mu^T\mu_i[j]))$ 
        \ENDFOR
    \ENDFOR
    \STATE return  $\{\text{ct}(\text{Cov}_i[j])| 0\leq i,j < k\}$       
    \end{algorithmic}
\end{algorithm}

\begin{theorem}
\label{thm:CovCplx}
The multiplication depth of Algorithm \ref{alg:Cov} is $2+\text{max}(\lceil\frac{ (n-1)}{n_Z}\rceil,\lceil \frac{(n-1)*n}{n_T}\rceil)+2$, and the time complexity is dominated by the number of matrix multiplications, which is $O(\lceil \frac{k^2}{2} \rceil \cdot \lceil s/n \rceil)$.
\end{theorem}
\begin{proof}
    We need 2 multiplication depths for matrix transpositions, 2 depths for column shiftings and ciphertext multiplications in the DMatrixMult function, and $k=max(k_1,k_2)$ depths for the linear transformations $Z$ and $T$, where $k_1=\lceil\frac{ (n-1)}{n_Z}\rceil$ is the number of depth used for decomposed $Z$, and $k_2=\lceil \frac{(n-1)\cdot n}{n_T}\rceil$ is the number of depth used for decomposed $T$.
\end{proof}

\section{Privacy-Preserving PCA using Homorphic PowerMethod}

We present the homomorphic evaluation circuit of the PowerMethod algorithm in the first subsection and demonstrate the process of performing privacy-preserving PCA utilizing our homomorphic PowerMethod circuit in the second subsection.

\subsection{Homomorphic PowerMethod Circuit}

We divide the demonstration of the homomorphic PowerMethod into three pivotal components: (i) the continuous transformation of the covariance matrix, (ii) the normalization of the approximate eigenvector, and (iii) the computation of the eigenvalue corresponding to the eigenvector. 

\subsubsection{Covariance Matrix Transformation}

For a column vector (or row vector) with $t$ components ($t$ is denoted as the number of features in Section \ref{sec: Ciphertext Packing}), we replicate it into a column (or row) of submatrices using the following approach:
\begin{enumerate}
    \item Divide the vector evenly into $k$ segments where each segment contains $n$ entries. The zero padding strategy is needed and should be consistent with the one used in ciphertext packing when $t$ is not divisible by $k$. 
    \item Horizontally (or vertically) replicate these $k$ segments of column vectors (or row vectors) $n$ times to construct $k$ submatrices. 
\end{enumerate}
The continuous linear transformation will be alternately applied to the horizontally replicated vector and the vertically replicated vector. Let $V_{t-1}$ be a set of sub-matrices representing a replicated column vector $\mathbf{v}_{t-1}$. For $0\leq i < k$, the covariance matrix transformation $\mathbf{v}_t=\text{Cov}\mathbf{v}_{t-1}$ can be computed as:
\begin{align}
\begin{split}
V_t[i]&=\text{Aggregate}\left(\sum_{0\leq j <k} (\text{Cov}_j[i]\odot V_{t-1}[j]^T)^T; 0\right) \\
&= \text{Aggregate}\left(\sum_{0\leq j <k} \text{Cov}_j[i]^T\odot V_{t-1}[j] ; 0\right).
\end{split}
\end{align}
Here, $V_t=\{V_t[i]|0\leq i<k\}$ represent the replicated row vector $v_t^T$. If $V_{t-1}$ represents a replicated row vector $\mathbf{v}^T_{t-1}$, then for $0\leq i < k$, the:  covariance matrix transformation $\mathbf{v}_t^T=\mathbf{v}^T_{t-1}\text{Cov}$ can be computed as:
\begin{align}
\begin{split}
    V_t[i]&=\text{Aggregate}\left(\sum_{0\leq j <k} V_{t-1}[j] \odot  \text{Cov}_i[j]^T; 1\right), \\
\end{split}
\end{align}
where $V_t$ represents the replicated column vector $\mathbf{v}_t$. This equation holds because the covariance matrix $\text{Cov}$ is symmetric, i.e., $\text{Cov}_j[i]^T = \text{Cov}_i[j]$. Now We can perform the covariance matrix transformation on a vector continuously using these two equations alternately. 

%From this, we can see that we can always perform the continuous linear transformation of the covariance matrix with Hadamard multiplication, summation, aggregation, and intermediate vectors alternating between row and column form.

The entire process of the continuous covariance matrix transformation can be described as the following. At the onset of the PowerMethod iteration, we initially select a row vector as the approximate eigenvector, denoted as $v_0$. We replicate it into a row of submatrices denoted as $V_{0}$. The set of submatrices representing the approximate eigenvector in the $t$-th iteration is computed as:
\begin{align}
\begin{split}
V_t[i]=&\text{Aggregate}\left(\sum_{0\leq j <k} V_{t-1}[j] \odot  \text{Cov}_j[i];t+1\text{ mod }2 \right)\cdot(t+1\text{ mod }2)+ \\
& \text{Aggregate}\left(\sum_{0\leq j <k} V_{t-1}[j] \odot  \text{Cov}_i[j];t\text{ mod }2 \right)\cdot(t\text{ mod }2).    
\end{split}
\end{align}

\subsubsection{Vector normalization}
Vector normalization should be performed in each iteration of the PowerMethod to the length of the eigenvectors from continually increasing or decreasing, which could lead to overflow. Normalizing a vector $\mathbf{v}$ is done by applying a scaling factor of $\frac{1}{||\mathbf{v}||_2}$ to each component of the vector. Specifically, the scaling factor is obtained by evaluating the InvSRT function at the point $\left< \mathbf{v},\mathbf{v} \right>$. We mentioned that we can only use the iterative InvSRT algorithm to approximate the effect of InvSRT in the homomorphic encryption scenario (see Section \ref{sec: InvSRT}). We summarize the following factors that significantly influence the accuracy of the iterative InvSRT algorithm: 

%In plaintext computation scenarios, it is customary to scale the approximate eigenvector results obtained in each round of the PowerMethod iteration and employ the scaled results as inputs for the next iteration. This is done to prevent the length of the eigenvectors from continually increasing or decreasing, which could lead to overflow. It is worth noting that the scaling performed during the iteration does not affect the convergence speed of the iteration. In the context of homomorphic encryption, where the actual values of the eigenvectors are not visible, it is difficult to scale the approximate eigenvectors using a fixed (set of) constant factors. Therefore, we adopt a normalization scaling approach, which ensures that the norm of the eigenvectors is kept within a reasonable range after each scaling. Given that current homomorphic encryption schemes generally do not support the evaluation of non-polynomial functions, achieving normalization necessitates approximating the InvSRT function within a certain range using an iterative approach. Based on existing research on the approximation of the InvSRT function using iterative algorithms in the homomorphic setting, we have summarized the following factors that significantly influence the accuracy of the approximation:
\begin{enumerate}
    \item Initial approximation: Given a fixed number of iterations, a closer initial approximation to the exact value yields more accurate results from the iterative algorithm. Conversely, an initial estimate that deviates significantly from the exact value may prevent the iterative algorithm from converging. Since the function $\frac{1}{\sqrt{x}}$ exhibits significant changes in slope on both sides of 1, resembling an "L" shape, determining an initial approximation that simultaneously approximates the exact value on both sides of 1 using a simple function (such as a low-degree polynomial) may be insufficient. However, constructing an initial approximation evaluation method that closely approximates the exact value on both sides of the evaluation interval can itself be a complex task. For example, approaches such as the piecewise function used in \cite{panda2022polynomial} or the rational polynomial function discussed in \cite{qu4258571improvements} yield initial estimates that are closer to the exact value. However, even these functions require an iterative approximation in the homomorphic setting.
    
    \item Number of iterations: The more iterations performed, the closer the output of the iterative InvSqrt algorithm will be to the exact value.
\end{enumerate}
To set up a good iterative InvSRT algorithm for the vector normalization in the homomorphic PowerMethod, our first step is to determine the upper bound of the norm of the approximate eigenvectors in order to establish the evaluation interval for the iterative InvSRT. Subsequently, we propose a Lazy Normalization strategy for the approximate eigenvector normalization in the PowerMethod. This strategy tolerates imprecise outputs of the iterative InvSRT algorithm within the input interval (0,1] during some rounds of the PowerMethod. However, it aims to achieve highly accurate outputs in the final iteration to ensure that the resulting eigenvectors are normalized.

%Faced with these factors, our first step is to determine the upper bound of the norm of the approximate eigenvectors in order to establish the evaluation range for the square root reciprocal iteration algorithm. Subsequently, we propose a Lazy Normalization strategy for scaling the approximate eigenvectors. This strategy tolerates imprecise outputs of the square root reciprocal iteration algorithm within the interval (0,1] during some rounds of the PowerMethod. However, it aims to achieve highly accurate outputs in the final iteration to ensure that the resulting eigenvectors are normalized.

\paragraph{Estimating the Upper Bound of the Norm}
%\noindent \textbf{Estimating the Upper Bound of the Norm}

After the dataset has undergone normalization and centering, we assume that all of its units are bounded within $b > 0$. For each element in the covariance matrix $\text{Cov}$, we have the bound $|\text{Cov}_{ij}| = \frac{1}{n}|\left< \mathbf{x}_i,\mathbf{x}_j \right>| \leq b^2$, where $\mathbf{x}_i, 0 \leq i < d$ represents the vector composed of the $i$-th feature values of all samples in the dataset. Therefore, the Euclidean norm of vector $\text{Cov}_i$ is bounded by $|| \text{Cov}_i ||_2 \leq \sqrt{b^4 \cdot d}$. With this information, we can estimate the upper bound of the norm for an approximate eigenvector $\mathbf{y}$ produced by the covariance matrix transformation $\mathbf{y} = \text{Cov} \cdot \mathbf{v}$ with the following inequality:
\begin{align}
    ||\mathbf{y}||_2 = \sqrt{\sum_{i=1}^d|y_i|^2}=\sqrt{\sum_{i=1}^d (\text{Cov}_{i1}v_1+\text{Cov}_{i2}v_2+\dots+\text{Cov}_{id}v_d)^2}\leq b^2\cdot c\cdot d,
\end{align}
where the Euclidean norm of the input approximate eigenvector $\mathbf{v}$ is bounded by $c$. In practical applications, different feature columns $\mathbf{x}_i$ may have different bounds $b_i$. If the cloud service provider has legal access to these bounds, tighter upper bounds can be constructed. In this case, we have $\text{Cov}_{ij} = \frac{1}{n}|\left< \mathbf{x}_i,\mathbf{x}_j \right>| \leq b_i b_j$, $||\text{Cov}_i||2 \leq \sqrt{\sum_{0 \leq j < d} b_i^2 b_j^2}$, and the upper bound of the Euclidean norm of $\mathbf{y}$ can be expressed as:
\begin{align}
    ||\mathbf{y}||_2 = \sqrt{\sum_{i=1}^d (\text{Cov}_{i1}v_1+\text{Cov}_{i2}v_2+\dots+\text{Cov}_{id}v_d)^2}\leq \sqrt{\sum_{i=1}^d(\sum_{j=1}^d b_i^2b_j^2)\cdot c^2)}.
\end{align}

\paragraph{Lazy Normalization}
%\noindent \textbf{Lazy Normalization} 

Based on the upper bound of the Euclidean norm, we design a vector normalization method called \textit{Lazy Normalization} to scale the eigenvectors. Let $\mathbf{y}$ remain the approximate eigenvector obtained from the covariance matrix transformation $\mathbf{y}=\text{Cov}\cdot \mathbf{v}$, and denote the upper bound for its Euclidean norm as $B$. We then set the evaluation interval for the iterative InvSRT algorithm as $(0,B^2]$. Next, we perform an odd-degree Taylor expansion $T(x)$ of $f(x) = \frac{1}{\sqrt{x}}$ at the midpoint $B^2/2+1$ which serves as the initial approximation function for the iterative InvSRT. It has been proven that such a Taylor expansion allows the subsequent iterative InvSRT to converge within a finite number of iterations \cite{qu4258571improvements}. The iterative InvSRT then takes the initial approximation $T(\left< \mathbf{y},\mathbf{y}\right>)$ as input and iterate $\tau$ times to obtain a more accurate value of $\frac{1}{||\mathbf{y}||_2}$. The approximate eigenvector is then multiplied by the output of the iterative InvSRT to perform the normalization.

For any $x\in (0,b^4\cdot c^2\cdot d^2] $, it holds that $T(x) \leq \frac{1}{\sqrt{x}}$ \cite{qu4258571improvements}. Furthermore, the iterative InvSRT output is no greater than the ground truth of InvSRT (Theorem \ref{thm:itr_invsqrt}). Hence, we can estimate that both the Euclidean have an upper bound of $1$ after the normalization. This allows us to recompute the upper bound of the Euclidean norm for the next approximate eigenvector $\mathbf{y}' = \text{Cov}\cdot \mathbf{y}$. In particular, if the length of $\mathbf{v}$ is initially $1$, then the upper bound of $||\mathbf{y}'||_2$ remains unchanged and equal to the upper bound of $||\mathbf{y}||_2$.

The degree of $T(x)$ is minimized to prevent generating extremely small high-degree coefficients that exceed the precision range of the homomorphic encryption scheme. This may lead to a result that the accuracy of $T(x)$ in the interval $(0,1)$ is lower compared to the interval $[1,B^2]$, due to the significant difference in slopes on both sides of $1$ for $\frac{1}{\sqrt{x}}$. Subsequent iterative InvSRT inherits this disparity of accuracy. It is impractical to achieve the same level of accuracy for the iterative InvSRT across the evaluation interval $(0,1)$ as in the interval $[1,B^2]$, since a significant number of iterations is required. The best we can do is to make an assumption that the covariance matrix transformation does not shrink the length of vectors by more than a factor of $S$. Based on this assumption, we can attempt to ensure that (i) the length of the eigenvector does not underflow during the PowerMethod iterations and (ii) the length of the final output eigenvector from the PowerMethod is sufficiently close to 1. % 要让 iterative InvSRT 在整个求值区间 $(0,1)$ 上的输出都具有和在求值区间 $[1,B^2]$ 上的输出相同的精确度需要付出大量的迭代，是不实际的。我们能做的是假设协方差矩阵变换对向量的缩小效应不超过一个界 $B'$，然后尝试在这个假设下使得 (i) 特征向量的长度在 PowerMethod 迭代过程中不会下溢出，(ii) PowerMethod 最终输出的特征向量的长度足够接近 1. 

Therefore, our strategy is as follows. Let $B'$ be a precision bound, $e$ be an error bound, and $t$ be the expected number of PowerMethod iterations. Denote $S$ as the contraction coefficient of the covariance matrix transformation, which means that any vector, when transformed by the covariance matrix, is reduced in length by at most a factor of $S$. We then simulate the PowerMethod iterations on a unit-length vector to obtain the best-fit number of iterative InvSRT iterations in each PowerMethod iteration:
\begin{enumerate}
    \item We start with an initial length $l$ of 1. In each PowerMethod iteration, we first update $l$ directly as $l \leftarrow l \cdot S$, without performing any actual covariance matrix transformation.
    \item Subsequently, we evaluate $T(l^2)$ and feed it into the iterative InvSRT. The resulting output, denoted as $S'$, must satisfy the following two conditions: (i) $(l\cdot S'\cdot S)^2$ should not underflow beyond the predetermined precision bound $B'$, and (ii) the value $1 - l\cdot S'$ must be smaller than $e$ in the final iteration of the PowerMethod. %The output of iterative InvSRT $S'$ should 满足以下两个条件：(i) $(l\cdot S'\cdot S)^2$ 不会下溢出某个预先设定的精度范围 (ii) 如果处于最后一轮 PowerMethod 迭代中，那么误差 $e= 1-l\cdot S'$ 必须要小于某个预先设定的范围。
    \item $l$ is updated as $l\leftarrow l\cdot S'$ in preparation for the subsequent round of PowerMethod iteration.
    \item  We maintain a record of the number of iterations employed by the iterative InvSRT during each iteration of the PowerMethod. In the actual homomorphic PowerMethod computation, we rely on these recorded iteration counts to perform the iterative InvSRT. %在真正的同态 PowerMethod 计算中，就按照记录的这些迭代轮数来进行 iterative InvSRT。 %记录下每一 PowerMethod iteration 中使用的 iterative InvSRT 迭代数。它们就是 Lazy Normalization 策略中调用 iterative InvSRT 时需要迭代的次数。
\end{enumerate}
Up to this point, our Lazy Normalization strategy can be uniquely characterized by its evaluation interval $(0, B^2]$, initial approximation function $T$, precision bound $B'$, error bound $e$, number of PowerMethod iterations $t$, and the number of iterative InvSRT iterations per PowerMethod iteration. Consequently, we can conclude that the Lazy Normalization strategy is capable of accommodating $t$ PowerMethod iterations that involve covariance matrix transformations with a contraction coefficient of $S$ while maintaining an error bound $e$ and a precision bound $B'$. %So far，我们的 Lazy Normalization 可以由其求值区间 $(0,B^2]$、 initial approximation function $T$、precision bound $B'$、error bound $e$、PowerMethod 迭代数 $t$，以及每轮 PowerMethod 迭代中 iterative InvSRT 的迭代数来唯一标识。We can then say that the Lazy Normalization strategy can tolerate $t$ PowerMethod iterations involving covariance matrix transformations with a contraction coefficient of $S$ up to an error of $e$.

\begin{theorem}
    \label{thm:itr_invsqrt}
        If Algorithm \ref{alg:Newton} is given an input $y_0 \leq \frac{1}{\sqrt{x_0}}$, then in any iteration, if $y_{i-1} \leq \frac{1}{\sqrt{x_0}}$ and $y_{i-1} > 0$, it follows that $y_{i-1} \leq y_i \leq x_0$.
\end{theorem}
\begin{proof}
    When $y_{i-1} \leq \frac{1}{\sqrt{x_0}}$ and $y_{i-1} > 0$, we have $0 < x_{0}y_{i-1}^2 \leq 1$. Therefore, $y_i \geq y_i$. On the other hand, set $y_{i-1} = x_0 - \epsilon$ where $0 \leq \epsilon < \frac{1}{\sqrt{x_0}}$. We can express $y_i$ as follows:
    \begin{align}
    \begin{split}
        y_i & = \frac{1}{2}(\frac{1}{\sqrt{x_0}}-\epsilon)\left(3-x_0\cdot \left( \frac{1}{\sqrt{x_0}}-\epsilon \right)^2  \right) \\
        & = \frac{1}{2}(\frac{1}{\sqrt{x_0}}-\epsilon)(2+2\sqrt{x_0}\epsilon - x_0\epsilon^2) \\ 
        & = \frac{1}{\sqrt{x_0}} - \frac{\sqrt{x_0}\epsilon^2}{2} - \sqrt{x_0}\epsilon^2 + \frac{x_0\epsilon^3}{2} \\ 
        & \leq \frac{1}{\sqrt{x_0}}.
    \end{split}
    \end{align}
\end{proof}

\begin{algorithm}[htbp]
    \renewcommand{\algorithmicrequire}{\textbf{Input:}}
    \renewcommand{\algorithmicensure}{\textbf{Output:}}
    \caption{TaylorInit}
    \label{alg:TaylorInit}
    \begin{algorithmic}[1]
    \REQUIRE
        $ct(x)$: ciphertext encrypting value $x$;
        $B$: upper bound for the iterative InvSRT evaluation interval;
        $o$: truncation order.
    \ENSURE
        $ct(x')$, where $x' = T(x)$, and $T(x)$ is the Taylor expansion polynomial of $\frac{1}{\sqrt{x}}$ at point $B/2+1$ with truncated order $o$.    
    \end{algorithmic}
\end{algorithm}

\begin{algorithm}[htbp]
    \renewcommand{\algorithmicrequire}{\textbf{Input:}}
    \renewcommand{\algorithmicensure}{\textbf{Output:}}
    \caption{IterativInvSRTbyNewton}
    \label{alg:InvSRTbyNewton}
    \begin{algorithmic}[1]
    \REQUIRE
        $ct(x)$: ciphertext encrypting point; $x$ %\mathbf{x}$,
        $ct(\text{Guess})$: initial approximation of iterative InvSRT at point $x$;
        $\tau$: number of iterations. 
    \ENSURE
        $ct(x'), {x' \approx \frac{1}{\sqrt{x}}}$    
    \STATE $ct(x')\leftarrow ct(\text{Guess})$  
    \FOR{i range 0 to $\tau-1$ }
	\STATE $ct(\text{Temp})\leftarrow \text{Mult}\left (\text{Mult}(ct(x'),ct(x')),\text{Mult}(ct(x'),ct(x))\right)$ 
	\STATE $ct(\text{Temp})\leftarrow \text{Sub}(ct(\text{Temp}),\text{PMult}(pt(3), ct(x')))$ 
        \STATE $ct(\text{Temp})\leftarrow \text{PMult}(pt(1/2),ct(\text{Temp}))$ 
	\STATE $ct(x')\leftarrow ct(\text{Temp})$
    \ENDFOR
    \STATE return $ct(x')$ 
    \end{algorithmic}
\end{algorithm}

\subsubsection{Eigenvalue Computation}

The process of computing the eigenvalue is similar to completing another round of Power Method iterations. Let $\mathbf{v}$ be the eigenvector obtained after $t$ iterations. The eigenvalue $\lambda$ is given by $\lambda = \frac{\left< \text{Cov}\cdot\mathbf{v},\mathbf{v} \right>}{\left< \mathbf{v}, \mathbf{v} \right>}$. We know that $\mathbf{v}$ has been normalized after $t$ iterations. Therefore, if we are confident enough in the accuracy of the normalization process, the eigenvalue computation can be simplified to $\lambda = \left< \text{Cov}\cdot\mathbf{v},\mathbf{v} \right>$. The process of $\text{Cov}\cdot \mathbf{v}$ is exactly the same as the first half of a new round of Power Method iterations. However, $\text{Cov}\cdot \mathbf{v}$ and $\mathbf{v}$ have different orientations, i.e., if one is a column vector, the other is a row vector, and vice versa. Therefore, when performing the dot product, we need to flip the orientation of $\mathbf{v}$ to match that of $\text{Cov}\cdot \mathbf{v}$. The axis flipping algorithm (AxisFlipping) is provided in Algorithm (\ref{alg:AxisFlip}). Next, computing $\left< \text{Cov}\cdot\mathbf{v},\mathbf{v} \right>$ only requires $k$ ciphertext multiplications and one aggregation operation. We can incorporate the eigenvalue computation into the PowerMethod iteration, specifically in the last round of the Power Method (total $t+1$ rounds), where the iteration is used to compute the eigenvalue instead of updating the eigenvector. Note that we need to retain $\mathbf{v}^T$ to compute the covariance matrix after the eigenvalue transformation.

\begin{algorithm}[htbp]
    \renewcommand{\algorithmicrequire}{\textbf{Input:}}
    \renewcommand{\algorithmicensure}{\textbf{Output:}}
    \caption{AxisFlipping}
    \label{alg:AxisFlip}
    \begin{algorithmic}[1]
    \REQUIRE \; \\
        $\{ct(V[i])|0\leq i < k\}$: a set of encrypted matrices representing one replicated vector; \\
        $\text{axis}$: the axis of the vector. \\
    \ENSURE \; \\
        $\{ct(V[i]^T)|0\leq i < k\}$   
    \STATE $\text{Mask} \leftarrow I$ 
    \FOR{i range 0 to $k-1$ }
	\STATE $ct(V[i]^T) \leftarrow \text{PMult}(\text{Mask},ct(V[i]))$ 
	\STATE $ct(V[i]^T)\leftarrow\text{Aggregate}(ct(V[i]^T);\text{axis}+1 \text{ mod }2)$ 
    \ENDFOR
    \RETURN $\{ct(V[i]^T)|0\leq i < k\}$ 
    \end{algorithmic}
\end{algorithm}

So far, we have discussed the three modules of the Power Method algorithm, namely the linear transformation of the covariance matrix, normalization of the eigenvectors, and computation of the eigenvalues. These modules cover all the processing details of the Power Method, and the complete process is provided in Algorithm \ref{alg:PowerMethod}, where the invoked iterative InvSRT algorithm is implemented by Newton's method with a Taylor expansion as the initial approximation (see Algorithm \ref{alg:TaylorInit} and \ref{alg:InvSRTbyNewton}). We analyze the multiplication depth and the computation complexity of Algorithm \ref{alg:PowerMethod} in Theorem \ref{thm:PMCplx}. % Notice that our multiplication depth complexity outperforms the result $O(l_P\cdot (\log(n)+\log{(order)}+z\cdot l_N+2))$ complexity reported in \cite{panda2021principal} ($z$ is a constant no smaller than 3), and our rotation complexity is also more favorable than 
% We can observe that it takes 3 levels to complete one iteration of Newton's method, and the Taylor expansion requires approximately $\log{(order)}$ levels. Hence, the multiplication depth of Algorithm \ref{alg:PowerMethod} is $O(l_P\cdot (3+\log{(order)}+3\cdot l_N+1)+2)$ which is a more favorable asymptotic complexity compared to the $O(l_P\cdot (\log(n)+\log{(order)}+z\cdot l_N+2))$ complexity reported in \cite{panda2021principal} ($z$ is a constant no smaller than 3). The computational complexity of Algorithm \ref{alg:PowerMethod} is dominated by the number of ciphertext multiplications, which is $O(l_P\cdot (k^2+k+\log{(order)}+3\cdot l_N)$, and the number of ciphertext rotations, which is $O(\log(n)\cdot (2\cdot l_P+1))$.

\begin{theorem}
\label{thm:PMCplx}
    The multiplication depth of Algorithm \ref{alg:PowerMethod} is $O(l_P\cdot (3+\text{log(order)}+3\cdot l_N+1)+3)$, and the time complexity is dominated by $O((l_P+1)\cdot (k^2+k+c+3\cdot l_N+1)-3\cdot l_N-1)$ ciphertext multiplications and $O(\frac{3k+3}{2}\log{n}\cdot (l_P+1))$ ciphertext rotations, where $c$ is the number of ciphertext multiplications needed for the Taylor polynomial evaluation.
\end{theorem}
\begin{proof}
    In each round of the PowerMethod, a maximum of 3 multiplication depths, $k^2+k$ ciphertext multiplications, and $(2k+1)\cdot \log{n}$ (or $(k+2)\cdot \log{n})$ ciphertext rotations are required to accomplish the covariance matrix transformation and one inner product. Subsequently, $\log{(order)}$ depths and $c$ ciphertext multiplications are needed for the Taylor Expansion. Once we enter Newton's Method iterative algorithm, each iteration necessitates 3 depths and 3 ciphertext multiplications. Lastly, a single multiplication depth is required to perform the ciphertext multiplication for scaling the approximate eigenvector. After conducting $l_P$ rounds of PowerMethod iterations, an additional covariance matrix transformation and inner product are required to complete the eigenvalue calculation.
    %每一轮 PowerMethod 中，我们首先需要花费至多 3 乘法深度、$k^2+k$ 次密文乘法、$(2k+1)\cdot \log{n}$（或$(k+2)\cdot \log{n})$）次密文旋转来完成 Covariance matrix transformation 与近似特征向量的内积操作，然后我们需要 $\log{order}$ 深度、$c$ 次密文乘法来进行泰勒展开的求值，进入到 Newton's Method 迭代算法以后我们需要 3 层级、3 密文乘法来完成一次迭代，最后需要一乘法深度来完成用于缩放近似特征向量的密文乘法。在进行了 $l_P$ 次 PowerMehto 迭代以后，我们还需要进行一次 Covariance matrix transformation 与一次向量内积来完成特征值计算。
\end{proof}

\begin{algorithm}[tp]
    \renewcommand{\algorithmicrequire}{\textbf{Input:}}
    \renewcommand{\algorithmicensure}{\textbf{Output:}}
    \caption{Homomorphic EigenShift (HEigenShift)}
    \label{alg:EigenShift}
    \begin{algorithmic}[1]
    \REQUIRE \; \\
        $ct(\lambda)$,$\{ct(V[i])|0\leq i < k\}$: the dominant eigenvalue and the corresponding eigenvector;\\
        $\{ct(V[i]^T)|0\leq i < k\}$: the transposed eigenvector;\\
        $\{\text{ct}(\text{Cov}_i[j])| 0\leq i,j < k\}$: the covariance matrix; 
    \ENSURE \; \\
         $\{\text{ct}(\text{SCov}_i[j])| 0\leq i,j < k\}$: the $1$-shifted version of $\text{Cov}$.
    \FOR{ $0\leq j <k$} 
	\FOR{ $0\leq i <k$}
		\STATE $\text{ct}(\lambda\cdot V^TV_i[j])\leftarrow \text{Mul}(ct(\lambda),\text{Mul}(\text{ct}(V[j]^T),\text{ct}(V[i])))$
		\STATE $\text{ct}(\text{SCov}_i[j]) \leftarrow \text{Sub}(ct(\text{Cov}_i[j]) ,ct(\lambda\cdot V^TV_i[j]))$
        \ENDFOR
    \ENDFOR
    \STATE return  $\{\text{ct}(\text{SCov}_i[j])| 0\leq i,j < k\}$ 
    \end{algorithmic}
\end{algorithm}

\begin{algorithm}
    \renewcommand{\algorithmicrequire}{\textbf{Input:}}
    \renewcommand{\algorithmicensure}{\textbf{Output:}}
    \caption{Homomorphic PowerMethod(HPowerMethod)}
    \label{alg:PowerMethod}
    \begin{algorithmic}[1]
    \REQUIRE \; \\
        $\{\text{ct}(\text{Cov}_i[j])| 0\leq i,j < k\}$: the encrypted covariance matrix; \\
        $l_P$: iterations for the PowerMethod;\\        
        $l_N$: iterations for the InvSRTByNewton;\\ 
        $\{ct(V_0[i])|0\leq i < k\}$: encrypted submatrices representing the initial selection for the approximated eigenvector;\\ 
        $\text{axis}$: the axis of the current approximate eigenvectors, where $\text{axis}=0$ represents row vectors with all rows being equal and $\text{axis}=1$ represents column vectors with all columns being equal;\\
        $B$: upper bound for the iterative InvSRT evaluation interval;\\
        $order$: truncation order for the Taylor initialization.\\
    \ENSURE \; \\
        $ct(\lambda)$, $\{ct(V_{l_P-1}[i]^T)|0\leq i < k\}$, $\{ct(V_{l_P-1}[i])|0\leq i < k\}$: Dominant eigenvalue of the covariance matrix, the corresponding eigenvector, along with its transposition.
    \FOR{ $t$ range 0 to $l_P$}
        \FOR{i range 0 to k-1}
		\STATE $ct(\text{Sum}) \leftarrow ct(\mathbf{0})$ 
		\FOR{ $j=0;j<k;j++$}
			\IF{ $\text{axis}==0$} 
				\STATE $ct(\text{Sum})\leftarrow \text{Add}(\text{Mult}(ct(\text{Cov}_j[i]),ct(V_t[i])),ct(\text{Sum}))$ 
			\ELSIF{$\text{axis} == 1$} 
				\STATE $ct(\text{Sum})\leftarrow \text{Add}(\text{Mult}(ct(\text{Cov}_i[j]),ct(V_t[i])),ct(\text{Sum}))$
                \ENDIF
            \ENDFOR
		\STATE $ct(V_{t+1}[i])\leftarrow \text{Aggregate}(ct(\text{Sum});\text{axis}+1 \text{ mod }2)$ 
        \ENDFOR
	\STATE $\text{axis} \leftarrow \text{axis} + 1 \text{ mod } 2$
	\STATE $ct(\text{Sum})\leftarrow ct(\mathbf{0})$ 
	\FOR{ i range 0 to k-1} 
		\STATE $ct(\text{Sum})\leftarrow \text{Add}(\text{Mult}(ct(V_{t+1}[i]),ct(V_{t+1}[i])),ct(\text{Sum}))$ 
	\ENDFOR
        \STATE $ct(\text{Sum})\leftarrow \text{Aggregate}(ct(\text{Sum});\text{axis}+1 \text{ mod }2)$ 
	\IF{ $n <  l_P$}
		\STATE $ct(\text{Guess})\leftarrow \text{TaylorInit}(ct(\text{Sum}),B,order)$ 
		\STATE $ct(\text{Sum})\leftarrow \text{InvSqrtByIteration}(ct(\text{Guess}),ct\text{Sum},l_N)$ 
	    \FOR{i range 0 to k-1 }
		      \STATE $ct(V_{t+1}[i])\leftarrow\text{Mult}(ct(V_{t+1}[i]),ct(\text{Sum}))$
            \ENDFOR
        \ENDIF	
        \IF{$t =  l_P$}
		\STATE $\{ct(V_t[i]^T)|0\leq i < k\}\leftarrow \text{AxisFlipping}( \{ct(V_t[i])|0\leq i < k\},\text{axis}+1\text{ mod }2 )$
		\STATE $ct(\text{Sum}) \leftarrow ct(\mathbf{0})$ 
		  \FOR{i range 0 to k-1}
			\STATE $ct(\text{Sum})\leftarrow \text{Add}(\text{Mult}(ct(V_{t+1}[i]),ct(V_t[i])),ct(\text{Sum}))$ 
            \ENDFOR		
            \STATE $ct(\lambda)\leftarrow \text{Aggregate}(ct(\text{Sum});\text{axis}+1\text{ mod }2)$ 
        \ENDIF
    \ENDFOR
    \STATE return $ct(\lambda)$, $\{ct(V_{l_P-1}[i]^T)|0\leq i < k\}$, $\{ct(V_{l_P-1}[i])|0\leq i < k\}$
    \end{algorithmic}
\end{algorithm}

\subsection{Privacy-Preserving PCA}

The final step is to combine the Power Method with the EigenShift algorithm to compute multiple dominant eigenvectors and their corresponding eigenvalues.  After each execution of the Power Method, we replace the covariance matrix with its $1$ shifted version, thereby ensuring that for $1\leq k$, the $k$-th round of the Power Method yields the $k$-th dominant eigenvector of the original covariance matrix of the dataset. The implementation of the homomorphic $1$-th EigenShift algorithm is presented in Algorithm \ref{alg:EigenShift}. It utilizes the dominant eigenvalue $\lambda$ and the corresponding eigenvector and its transpose $\mathbf{v},\mathbf{v}^T$ obtained in the previous Power Method iteration to update the covariance matrix as $\text{Cov}' = \text{Cov} - \lambda\cdot \mathbf{v}^T \mathbf{v}$. It can be proven that the $1$-shifted covariance matrix always has a smaller amplification effect on the vectors than its non-shifted version (see Appendix \ref{apd:B} for details). Thereby, we do not need to recompute the upper bound for the norm of vectors transformed by the shifted covariance matrix in each execution of PowerMethod. The entire process of the privacy-preserving PCA algorithm is provided in Algorithm \ref{alg:PPPCA}. 

%The implementation of the EigenShift algorithm is presented in Algorithm \ref{alg:EigenShift}. It utilizes the eigenvalue $\lambda$ and the eigenvector and its transpose $\mathbf{v},\mathbf{v}^T$ obtained in the previous Power Method iteration to update the covariance matrix as $\text{Cov}' = \text{Cov} - \lambda\cdot \mathbf{v}^T \mathbf{v}$. After each execution of the Power Method, we replace the covariance matrix with its eigen-shifted version, thereby ensuring that for $1\leq k$, the $k$-th round of the Power Method yields the $k$-th dominant eigenvector of the original covariance matrix. It can be demonstrated that each eigen-shifted covariance matrix amplifies vectors of identical length no more than the amplification effect of the original covariance matrix. Consequently, there is no necessity to recalculate the upper bound for the length of vectors transformed by the eigen-shifted covariance matrix.  %我们每完成一次 PowerMethod，就将协方差矩阵替换为它的 eigen-shifted 版本，使得对于 $1\leq k$，第 $k$ 轮 PowerMethod 计算得到的特征向量是原初协方差矩阵的第 $k$ dominant eigenvector. 可以被证明的是，每个 eigen-shifted 协方差矩阵对任意长度相同的向量的放大效果都不会超过原初协方差矩阵的放大效果，因此我们不需要重新计算经过 eigen-shifted 协方差矩阵变换后的向量的长度上界。

One more important consideration is the need for modulus refreshes when the multiplication depth of our algorithm exceeds the maximum multiplication depth supported by the underlying CKKS scheme.  %最后一件重要的事情：我们需要在算法中进行密文的模数刷新，因为我们的算法的乘法深度大概率会超过 CKKS 所能支持的最大乘法深度。
%the underlying CKKS scheme used in our algorithm only supports homomorphic operations up to a limited multiplication depth. This implies that we need to refresh the ciphertexts that have reached the maximum multiplication depth to higher modulus levels in order to support further homomorphic computation. 
It is always necessary to carefully select appropriate positions in the algorithm to perform modulus refreshes in order to minimize the overall number of refresh operations. We defer the discussion on the timing of the modulus refresh operations in our privacy-preserving PCA algorithm to Appendix \ref{apd:C} for further details. %我们把对 The timing of the modulus refresh operations in our privacy-preserving PCA algorithm 的讨论留到 Appendix \ref{apd:B} 进行.

\begin{algorithm}
    \renewcommand{\algorithmicrequire}{\textbf{Input:}}
    \renewcommand{\algorithmicensure}{\textbf{Output:}}
    \caption{Privacy-Preserving PCA}
    \label{alg:PPPCA}
    \begin{algorithmic}[1]
    \REQUIRE \; \\
        ${ \text{ct}(X_j[\ell]) | 0\leq j < k, 0\leq \ell < \lceil s/n \rceil }$: partitioned and encrypted dataset;\\ 
        $n_Z$: maximum diagonal index for transformation $Z$;\\
        $n_T$: maximum diagonal index $n_T$ for transformation $T$; \\
        $l_E$: expected number of eigenvectors or principal components. 
        $l_P$: iterations for the PowerMethod; \\ 
        $l_N$: iterations for the InvSRTByNewton; \\
        $B$: upper bound for the iterative InvSRT evaluation interval; \\
        $order$: truncation order for the Taylor initialization.        
    \ENSURE \; \\
        $\{\{ ct(V[i]) |0\leq i <k \}_j |0\leq j <l_E\}$ , $\{ct(\lambda_j)|0\leq j <l_E\}$: $l_E$ dominant eigenvectors and the corresponding eigenvalues.
    \STATE $\{ct(\text{Cov}_j[\ell])\}\leftarrow \text{HCovMtrx}({ \text{ct}(X_j[\ell]) | 0\leq j < k, 0\leq \ell < \lceil s/n \rceil },n_Z,n_T)$
    \FOR{$m:=0; m<l_E;m++$}
	\STATE $\{ct(V_0[i])|0\leq i < k\}\leftarrow \text{Randomly Generate with }\ell_2=1$
	\STATE $\text{axis}\leftarrow 1$ 
	\STATE $ct(\lambda_m),\{ct(V)\}_m,\{ct(V^T)\}_m,\leftarrow\text{HPowerMethod}(\{ct(\text{Cov}_j[\ell])\},l_P,l_N,\{ct(V_0[i])\},\text{axis},B,\text{order})$ 
	\STATE $\{ct(\text{Cov}_j[\ell])\}\leftarrow \text{EigenShift}(ct(\lambda_m),\{ct(V)\}_m,\{ct(V^T)\}_m,\{ct(\text{Cov}_j[\ell])\})$
    \ENDFOR
    \STATE return $\{\{ ct(V[i]) |0\leq i <k \}_j |0\leq j <l_E\}$ , $\{ct(\lambda_j)|0\leq j <l_E\}$ 
    \end{algorithmic}
\end{algorithm}

\section{Implementation}
We implement all the proposed algorithms described in the previous sections using the Lattigo library V4.1.0 \cite{lattigo}, which provides an implementation of the full-RNS CKKS scheme in Golang. The implementation of all the algorithms can be found in \cite{Ma2023}. Our experiments are conducted on a machine running Windows 10, equipped with an AMD Ryzen 7 7735HS (3.19 GHz) processor and 32GB of memory.

First, we present the performance results of the $Z$ and $T$ linear transformations using our diagonal convergence decomposition technique proposed in Section 4. Then, we proceed to report the performance of our privacy-preserving PCA and compare it with the state-of-art work by Panda. \cite{panda2021principal}. Their work is the most relevant to our scenario, as it also focuses on designing privacy-preserving PCA based on the CKKS scheme in a cloud computing service setting.

%All Our experiments were conducted using the full RNS CKKS approximate numerical homomorphic encryption scheme implemented in Lattigo V4.1.0. All experimental data were obtained using a machine equipped with an AMD Ryzen 7 7735HS with Radeon Graphics (3.19 GHz) processor, running on a 64-bit Windows 10 operating system with 32 GB of memory. 我们首先给出第 4 节中提出的 Diagonal Convergence Decomposition 对 double-hoisting 矩阵乘法的优化效果，然后对 Privacy-Preserving PCA using PowerMethod 的实现细节与表现进行展示。

\subsection{Performance of $Z$ and $T$ Using Diagonal Convergence Decomposition}

We conduct tests on the space optimization effect of our Diagonal Convergence Decomposition on the $128\times 128$ double-hoisting matrix multiplication. We evaluate the performance of the $Z$ and $T$ transformations decomposed at different maximum expected diagonal indices and compare them with the original double-hoisting ones with no decomposition applied (see Table \ref{tb:ZTData}). The evaluation is conducted on a concrete CKKS instance that offers a maximum available modulus level of $14$ with parameters: $\log{N} = 15, \log{b}=760,\sigma=3.2$ (see Appendix \ref{apd:A} for detailed parameter explanation). A baseline level of $14-4=10$ is set, and the initial modulus levels for all tested transformations are assigned appropriately to ensure that all transformation results fall down to the baseline level. 
%since We primarily focus on the reduction of the key space for the transformations $Z$ and $T$ achieved by the decomposition.
From Table \ref{tb:ZTData}, we can observe that for a fixed inner loop count, a smaller expected maximum diagonal index leads to a greater reduction in key space. In the decomposition with an expected maximum diagonal index four times smaller than 128, we successfully reduce the rotation key space of $Z$ by 58\% for an inner loop of 8 and by 35\% for an inner loop of 16. 

Additionally, our decomposition improve the computation speed of $Z$ and $T$. This is because, Within the matrix transformation product chain obtained after decomposition, the transformations that are initially performed at high levels (the ones on the right-hand side of the chain) require only a small amount of complexity, and the main complexity arises from the final matrix transformation in the chain (the one on the left-hand side), which is performed at the same level as the original scheme and has a smaller scale compared to the transformation that was not decomposed. Therefore, our decomposition scheme provides an optimized trade-off between space and speed, making the overall matrix multiplication scheme more flexible. % 我们用 $128\times 128$ 的矩阵乘法来测试我们的 Diagonal Convergence Decomposition 对 double-hoisting 矩阵乘法的空间优化效果。我们主要关注该 decomposition 对变换 $Z$ 和 $T$ 的密钥空间减少效果。通过选定一组 full-RNS CKKS 参数：$\text{log}QP=760, \text{log}N=15, \text{log(Slots)}=14$，其中可用的最高模数层级为 $14$ 层，我们对配备 Diagona Convergence Decomposition 的变换 $Z$, $T$ 在不同的最大期望对角线序号下的表现进行了测试，并将其与原 double-hoisting 矩阵乘法方案进行比较（见表 \ref{tb:ZTData}）。我们测试了分别消耗2层和4层模数（对应分解后最大对角线序号是 128 的 1/2 和 1/4）的decomposition在$Z$ 和 $T$ 上的表现。为此我们设置了一个底线层级 $Maxlevel - 4 = 10$，并为测试的所有线性变换分别设置起始模数层级，使得他们在完成后都能跌落到相同的底线层级上。从中我们可以看出，对于固定的内层循环数，分解程度越深，对密钥空间的减少就越多，其中对于消耗 4 层级的分解，在外层循环为 8 时我们成功将 $Z$ 的旋转密钥空间降低了百分之 58，在外层循环为 16 时降低了百分之 35。 此外，我们的分解还提升了 $Z$ 和 $T$ 的计算速度，这是因为我们的分解只需要在初始高层级上进行少量的旋转运算后就快速跌落到低层级上，而在低层级上的线性变换的规模又小于原始方案的规模。因此我们的分解方案给出了进一步对空间与速度的取舍方案，使得整个矩阵乘法方案更具弹性。

% More specific experimental data is provided in Table (\ref{tb:ZTData}), using the Z and T transformations in a 128 x 128 matrix multiplication as an example. It shows the key size, maximum runtime space, and execution time for different internal loop counts and diagonal convergence decompositions with different expected maximum non-zero diagonal indexes. It should be noted that the internal loop count ($n_1$) is linearly related to the maximum runtime space of the linear transformation. From the table, we can observe that our diagonal convergence decomposition reduces the number of rotation operations required for the linear transformation as the maximum diagonal index of the decomposed matrices decreases. When maintaining a low internal loop count (e.g., $n_1=8$), the reduction in rotation key count is more than half, allowing for more instances of linear transformations to be parallelized within a fixed space. As a result, our approach achieves higher efficiency, lower runtime space requirements, and reduced rotation key count. Finally, we provide the complete process of ciphertext decomposition matrix multiplication with diagonal convergence decomposition in Algorithm \ref{alg7}, omitting the internal optimization details such as double-hoisting and hoisting, and presenting a high-level illustrative flow.

\begin{table}[tb]
%\centering
\ra{1.5}
    \resizebox{\textwidth}{!}{
    \tabcolsep=0.2cm
    \begin{tabular}{@{}cccccccc@{}}\toprule
    \textbf{Scheme} &\textbf{Matrix} & $n_1$ &   \textbf{MaxDiagNo.} &  \textbf{Rtk(MB)} & \textbf{Ct(MB)} & \textbf{time(s)} & \textbf{Depth} \\
    \midrule
    \text{dh-BSGS} & Z & 8  & 127 & 1755 & 53 & 2.004 & 1 \\
    %\hdashline
    \text{DCDmp} & Z & 8  & 64  & 1080 & 53 & 1.362 & 2 \\
%\hdashline
    \text{DCDmp} & Z & 8  & 32  & \textbf{720}  & 53 & 1.249 & 4 \\
\hdashline
    %\text{dh-Ours} & Z & 8  & 16  & 540  & 53 & 1.707 & 8 \\
    %\hline
    \text{dh-BSGS} & Z & 16 & 127 & 1395 & 113 & 1.286 & 1 \\
%\hdashline
    \text{DCDmp} & Z & 16 & 64  & 1080  & 113 & \textbf{1.094} & 2 \\
%\hdashline
    \text{DCDmp} & Z & 16 & 32  & 900  & 113 & 1.179 & 4 \\
\hline
    \text{dh-BSGS} & T & 8  & $127\cdot 128$ &  1035 & 53 & 1.049 & 1 \\
%\hdashline
    \text{DCDmp} & T & 8  & $64\cdot 128$ &  720 & 53 & 0.701 & 2 \\
%\hdashline
    \text{DCDmp} & T & 8  & $32\cdot 128$ &  \textbf{540} & 53 & 0.652 & 4 \\
\hdashline
    \text{dh-BSGS} & T & 16 & $127\cdot 128$ &  1035 & 113 & 0.725 & 1 \\
%\hdashline
    \text{DCDmp} & T & 16 & $64\cdot 128$ &  900   &  113 & \textbf{0.572} & 2 \\
%\hdashline
    \text{DCDmp} & T & 16 & $32\cdot 128$ &  810   & 113 & 0.613 & 4 \\
    \bottomrule
    \end{tabular}
    }
\caption{Run-time performance Comparison of LinTrans $Z,T$ between original double-hoisting BSGS(dh-BSGS) and our Diagonal Convergence Decomposed scheme(DCDmp). \textbf{Rtk(MB)} represents the total space required for rotation keys, while \textbf{Ct(MB)} represents the total space required for caching hoisted ciphertexts during linear transformations, both measured in MB (megabytes).}  \label{tb:ZTData}
\end{table}

\subsection{Privacy-preserving PCA}
\subsubsection{Communication Model}
Our Privacy-preserving PCA algorithm serves the cloud computing scenario. We describe the communication model on which our algorithm is based in the following.
\begin{itemize}
    \item Encryption of User Data: The user employs homomorphic encryption to encrypt their data, ensuring the confidentiality of the data, and transmits it to the cloud.
    \item Cloud-side Computation: Upon receiving the encrypted data from the user, the cloud service provider conducts the privacy-preserving PCA algorithm on their end, leveraging homomorphic encryption techniques to carry out computations without decrypting the user's data.
    \item Return of Computation Results: Once the computation is completed, the cloud service provider returns the results to the user in encrypted form. The user can then locally decrypt the ciphertext results using their private key to obtain the plaintext results of the computation.
\end{itemize}
Through this communication model, the user's data remains encrypted at all times, and the plaintext is only accessible after decryption on the user's end. 

\subsubsection{Accuracy Measurement}
We provide two metrics to measure the accuracy of our Privacy-Preserving PCA algorithm. The first metric is $R2(X)$, which represents the R2 score between the reconstructed dataset obtained from privacy-preserving PCA and the original dataset. It is demonstrated that a PCA algorithm is considered to have good accuracy if it achieves an $R2(X)$ score larger than 0.3 in \cite{panda2021principal}. % 0.7

The second metric is $R2(V)$, which quantifies the similarity between the principal components computed by privacy-preserving PCA and the ground truth principal components computed using the \textit{np.linalg.eig} function from the \textit{numpy} package in Python. A higher value of $R2(V)$, closer to 1, indicates that the principal components obtained from privacy-preserving PCA are more similar to the ground truth.

%The second metric is $R2(V)$, which directly represents the R2 score between the principal components computed by privacy-preserving PCA and the ground truth principal components computed by the \textit{np.linalg.eig} function in the \textit{numpy} package of Python. $R2(V)$ 的值越接近 1，说明隐私保护 PCA 计算得到的主成分越接近于ground truth。. We obtain the ground truth principal components by computing all the eigenvectors of the covariance matrix using the \textit{np.linalg.eig} function in the \textit{numpy} package of Python.
%我们提供两种用于度量我们的Privacy-Preserving PCA 算法的准确率的指标. 第一种度量指标是 $R2(X)$, It represents the R2 score between the reconstructed dataset obtained from privacy-preserving PCA and the original dataset, which 在 \cite{panda2021principal} 中被认为如果能取得处于 $0.3\sim 0.7$ 之间的 $R2(X)$ 值的 PCA 算法具有较好的准确率。第二种度量指标是 $R2(V)$，它直接代表了由 privacy-Preserving PCA 计算得到的主成分和明文下 ground truth 主成分之间的 $R2$ 分数，其中 ground truth 主成分我们使用 Python numpy 包中的 np.lialg.eig 函数来计算得到协方差矩阵的所有特征向量来得到。

\subsubsection{Experimental Setup}

\paragraph{Modulus Refresh Strategy Selection}

The modulus refresh operation in our privacy-preserving PCA algorithm can be theoretically achieved through the following two approaches based on the cloud service scenario.
\begin{enumerate}
    \item The interactive strategy: the cloud service provider sends the ciphertexts that have reached the maximum depth to the user. The user then decrypts these ciphertexts, re-encrypts them on the maximum modulus level, and sends them back to the cloud for further homomorphic computation. This approach introduces communication overhead and requires the user to listen to the cloud continuously.
    \item The non-interactive strategy: the cloud service provider performs bootstrapping locally on the ciphertexts. While several bootstrapping methods designed for CKKS have been proposed in recent years, their software implementations are still relatively time-consuming. Therefore, the number of modulus refreshes will become one of the dominant factors in the algorithm's time complexity when using this strategy.
    %The non-interactive strategy: cloud service provider performs bootstrapping locally on the ciphertexts. 虽然近年来许多为 CKKS 设计的 Bootstrapping 方法被提出，但他们的软件实现仍然是相对耗时的，因此使用这种模数刷新方法时，模数刷新的次数将成为主宰算法时间复杂度的因素之一。
\end{enumerate}

We use the interactive modulus refresh strategy throughout our experiment since our algorithm consumes a large number of modulus levels and does not yield satisfactory results when bootstrapping routine is utilized to achieve the non-iterative modulus refresh strategy. %FIXME：Create a 脚注。
However, it is worth noting that there exist hardware-accelerated solutions for bootstrapping that can significantly reduce its computational overhead. Our privacy-preserving PCA algorithm may become truly practical under the non-interactive modulus refresh strategy with such hardware optimizations.  \\ 

\paragraph{Parameterization for Underlying Homomorphic Encryption Scheme}

We set up a concrete CKKS instance with parameters $\log{N}=15, b=720,\sigma=3.2,\delta = 40$ as the underlying homomorphic encryption scheme of our privacy-preserving PCA algorithm (see Appendix \ref{apd:A} for the explanation of parameters), ensuring 128-bit security. This instance provides 13 levels of 40-bit moduli and achieves 20 bits of precision.\\

%We set up two concrete CKKS instances to test the performance of our privacy-preserving PCA under the two different modulus refresh strategies above. For the interactive modulus refresh strategy, we initialize a CKKS instance with the set of parameters $\log{N}=15$, $\log{Slots}=14,\log{QP}=720,\sigma=3.2$. This parameter set provides 13 levels of 40-bit moduli. For the non-interactive modulus refresh strategy, we choose the set of parameters $\log{N}=16,\log{Slots}=14,\log{QP}=1078+465,\sigma=3.2$ for the CKKS instance initialization, which is a recommended parameter set provided by the Lattigo library to conduct the bootstrapping. The value $1078$ represents the total bits of moduli used for the bootstrapping circuit, and the remaining $465$ bits provide $9$ valid levels of $45$-bit moduli. Both sets of parameters ensure 128-bit security.

\paragraph{Dataset Preprocessing} 

We conduct our privacy-preserving PCA algorithm on the MNIST \cite{deng2012mnist}, Fashion-MNIST \cite{xiao2017fashion}, and YALE \cite{belhumeur1997eigenfaces} datasets. The preprocessing of the three datasets, including cropping, is performed exactly as described in the work by Panda \cite{panda2021principal}. The images in the MNIST and Fashion-MNIST datasets are cropped to 16x16 pixels, and the Yale samples are converted to grayscale and cropped from 195 x 231 pixels to 16x16 pixels. In the subsequent experiments, we will extract a specific number of samples $s$ from these datasets to construct new datasets for testing our algorithm. All the datasets used in the experiments are divided into submatrices of size $128\times 128$ since our underlying homomorphic encryption scheme offers a plaintext space of $\mathbb{C}^{14}$ where each ciphertext can accommodate exactly one submatrix. \\
%按照我们在第 5.1 节介绍的密文打包策略，我们会使用 $128\times 128$ 的子矩阵来对这些数据集进行划分，这样一来 MNIST 和 Fashion-MNIST 都将由 $496\times 2$ 个子矩阵构成，而 Yale 数据集将由 $2\times 2$ 个子矩阵构成。

\paragraph{Parameterization for Covariance Matrix Computation} 

Our covariance matrix computation is designed to harness the advantages of parallel matrix multiplications. We employ a single thread to compute the covariance matrix in cases where the dataset size is no larger than $2\times 2$ sub-matrices. As for dataset sizes surpassing this threshold, we allocate seven threads to compute the corresponding covariance matrix. The $128\times 128$ matrix multiplications within the covariance matrix computation are implemented by our optimized matrix multiplication algorithm equipped with the hoisting techniques and a Diagonal Convergence Decomposition strategy with an expected maximum non-zero diagonal index of 32. This configuration reduces 35\% and 20\% in the requisite number of rotation keys for the $Z$ and $T$ transformations with an inner loop count of $16$, respectively. Such optimization mitigates the risk of memory saturation, diminishes the communication overhead between the user and the cloud, and decreases computational time. We also implement all three covariance matrix computation optimizations mentioned in Section \ref{sec: Cov}. Specifically, we group every 1400 rows of samples into one segment and apply optimizations 1 and 2 to each segment when the dataset contains more than 1400 rows. The covariance matrices computed from different segments are aggregated to obtain the final result.\\

%In order to achieve such parallelism without encountering memory saturation, we further implemented the Diagonal Convergence Decomposition strategy on the $128\times 128$ matrix multiplications with an expected maximum non-zero diagonal index of 1/4 of the original maximum index. This specific approach results in a reduction of 58\% and 47\% in the requisite number of rotation keys for the transformations of $Z$ and $T$, respectively. \\
%另外，我们的协方差矩阵计算需要用到多线程加速，在我们的设备上，分配给该计算的线程超过 7 后由于内存饱和其速度不再有显著提升，因此我们设定对于数据集规模不大于 $2\times 2$（以 $128\times 128$ 子矩阵作为单位）时，两线程足够计算协方差矩阵，而对于数据集规模超过这个规模的情况，我们采用 $7$ 线程计算协方差矩阵。

%对于 $128\times 128$ 规模的矩阵乘法，我们选择使用期望最大非零对角线序号为原最大序号的 1/4 的 Diagoanl Convergence Decomposition 策略，以将矩阵乘法中 $Z$ 和 $T$ 变换需要的密钥数量分别降低 58\% 和 47%.

\paragraph{Parameterization for Vector Normalisation}

In section 6.1.2, we estimate that performing a covariance matrix transformation on a vector of length $c$ results in a transformed vector with an upper bound length of $b^2\cdot c\cdot d$, where $b$ is the upper bound of all elements in the covariance matrix and $d$ is the number of features in the dataset. Since the upper bound length of the transformed vector will be used to set up the evaluation interval of the iterative InvSRT function, we should not let this upper bound be too big or it will be too expensive to evaluate the InvSRT function across the interval. Thereby, we scale all samples in the preprocessed datasets by a factor of 1/255, resulting in an upper bound $b=1$ for all elements in their covariance matrices.

Following the Lazy Normalization strategy, we set the evaluation interval of the iterative InvSRT function to $[0.1^{-7}, b^2\cdot c\cdot d]$ and use a linear Taylor expansion at point $2^{16}/2+1$ to obtain the initial approximation of the iterative InvSRT algorithm. We determine through experimentation that the iterative InvSRT requires a minimum of $k=12$ basic iterations to achieve the desired accuracy within the evaluation interval $[1,2^{16}]$. In order to make the Lazy Normalization strategy be capable of accommodating up to $l_P$ PowerMethod iterations that involve covariance matrix transformations with a contraction coefficient of $0.5$ while maintaining an error bound $e=1e-5$ and a precision bound $p=1e-3$, PowerMethod simulations are conducted to select the number of iterative InvSRT iterations in each PowerMethod iteration, for $l_P=4$ and $l_P =14$ (see Figure \ref{fig:VecN}). Consequently, we manage to employ an average of $14.5$ iterative InvSRT iterations per PowerMethod iteration for $l_P=4$, and an average of $13.8$ iterative InvSRT iterations per PowerMethod iteration for $l_P=14$.

\begin{figure}[htbp]
\centering
\begin{minipage}[t]{0.48\textwidth}
\centering
\includegraphics[width=8cm]{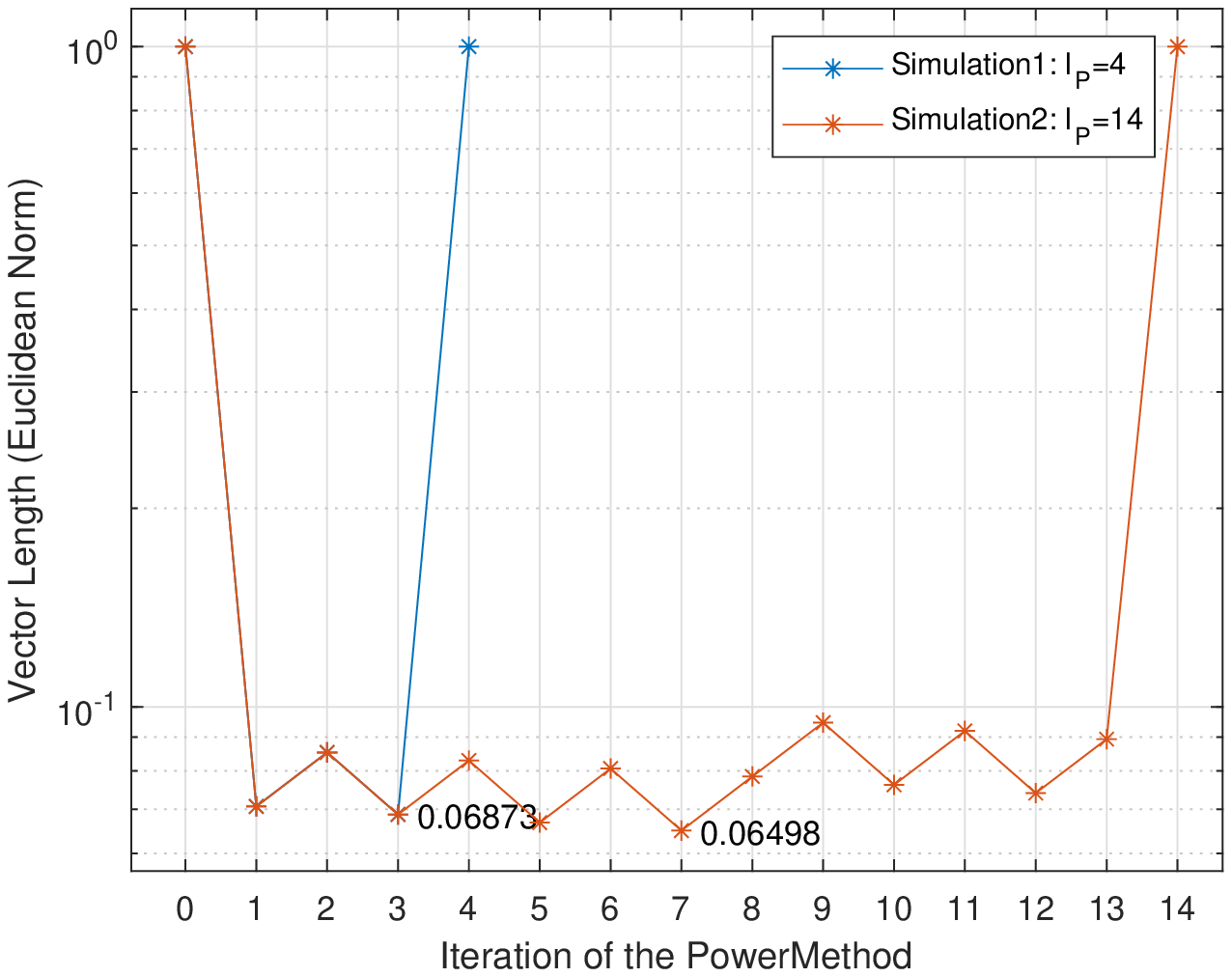}
%\caption{World Map}
\end{minipage}
\begin{minipage}[t]{0.48\textwidth}
\centering
\includegraphics[width=8cm]{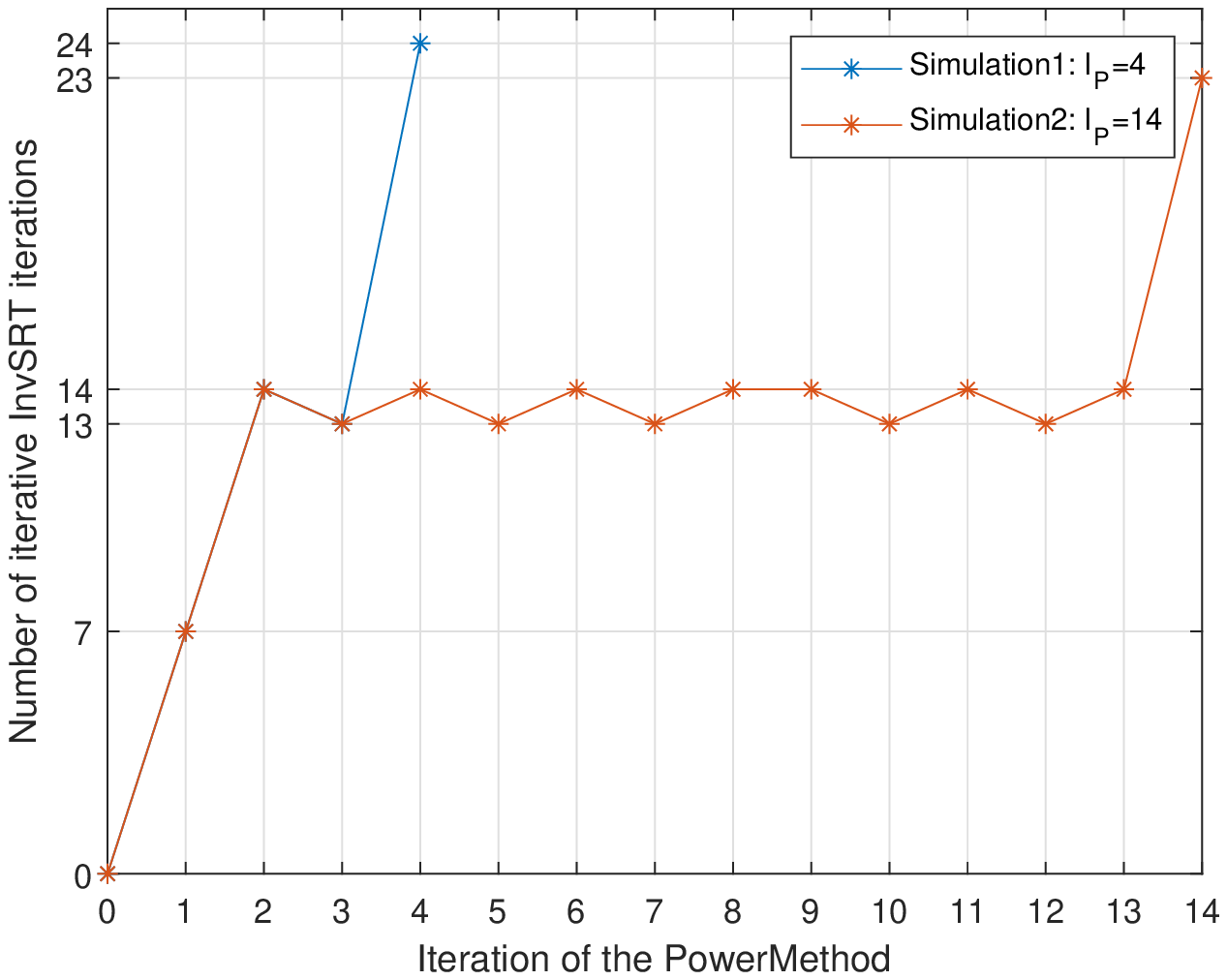}
%\caption{Concrete and Constructions}
\end{minipage}
\caption{PowerMethod Simulations for Vector Normalization Parameterization. The figures presented on the left and right respectively illustrate the evolution of vector lengths and the number of iterative InvSRT iterations per PowerMethod iteration in 4 and 14 rounds of simulations. The simulations are conducted with an error bound of $e=1e-5$ and a precision bound of $p=1e-3$. %左右两图分别展示了 $4$ 轮和 $14$ 轮 PowerMethod simulation 中向量的长度变化与每 PowerMethod iteration 所用的 iterative InvSRT 迭代数，对于误差上界 $e=1e-5$ 以及精度上界 $p=1e-3$。
}
\label{fig:VecN}
\end{figure}

\begin{table}[ht]
\ra{1.3}
    \resizebox{\textwidth}{!}{
    \tabcolsep=0.2cm
    \begin{tabular}{@{}llllllllll@{}} \toprule
    %\hline
$\text{Scheme}$ & $\text{DataSet}$ & $s$ &  $l_E$ & $l_P$ & $l_N$
& $time$(min) & \makecell[l]{$R2(X)$} & \makecell[c]{$R2(V)$}   & Lvs/PM \\ \midrule
%\hline
%1
$\text{\cite{panda2021principal}}$ & $\text{MNIST}$ & 200 & $4$ & $4$ & $6$  & $\text{5.21}$ & \makecell[l]{$0.448$} & $0.047$ &  $77$\\
%\hline
%2
$\text{Ours}$ & $\text{MNIST}$ & 200 & $4$ & $4$ & $14.5$  & $\uline{\text{0.56+2.16}}$ & \makecell[l]{$0.485$} & \makecell[l]{$0.174$} &  $61$ \\
%\hline
%3

$\text{\cite{panda2021principal}}$ & $\text{MNIST}$ & 200 & $4$ & $14$ & $6$   & ${\text{18.42}}$ & \makecell[l]{${0.490}$} & \makecell[l]{${0.196}$} & $77$ \\

$\text{Ours}$ & $\text{MNIST}$ & 200 & $4$ & $14$ & $13.8$   & ${\text{0.58+8.41}}$ & \makecell[l]{$\uline{0.491}$} & \makecell[l]{$\uline{0.349}$} & $54$ \\

%\hline
$\text{\cite{panda2021principal}}$ & $\text{MNIST}$ & 60000 & $8$ & $14$ & $6$ & $\text{1628}$ &  \makecell[l]{$0.306$} & \makecell[l]{$0.445$}  & $77$ \\

$\text{Ours}$ & $\text{MNIST}$ & 60000 & $8$ & $14$ & $13.8$  & $\uline{45+14}$ & \makecell[l]{$\uline{0.317}$} & \makecell[l]{$\uline{0.728}$}  & $54$ \\
\hline
%\hdashline
$\text{\cite{panda2021principal}}$ & $\text{F-MNIST}$ & 200 & $4$ & $4$ & $6$  & $\text{5.3}$ & \makecell[l]{$0.411$} & $0.471$ & $77$\\
%\hline
$\text{Ours}$ & $\text{F-MNIST}$ & 200 & $4$ & $4$ & $14.5$  & $\uline{{0.58+2.23}}$ & \makecell[l]{$\uline{0.478}$} & \makecell[l]{$\uline{0.650}$} & $61$\\ 
%\hline
$\text{\cite{panda2021principal}}$ & $\text{F-MNIST}$ & 60000 & $8$ & $14$ & $6$   & $\text{1628}$ & \makecell[l]{$0.508$} & \makecell[l]{$0.703$}  & $77$ \\

$\text{Ours}$ & $\text{F-MNIST}$ &  60000 & $8$ & $14$ & $13.8$  & $\uline{{45+14}}$ & \makecell[l]{$\uline{0.577}$} & \makecell[l]{$\uline{0.911}$} & $54$\\ 
\hline
%\hdashline
$\text{\cite{panda2021principal} }$ & $\text{Yale}$ & 165 & $6$ & $14$ & $6$ & ${27.5}$ & \makecell[l]{$0.665$} & $0.792$ & $77$\\
%\hline
$\text{Ours}$ & $\text{Yale}$ & 165  & $6$ & $14$ & $13.8$ & ${\uline{0.62+11}}$ & \makecell[l]{$\uline{0.674}$} & \makecell[l]{$\uline{0.917}$} & $54$ \\  \bottomrule
%\hline
    \end{tabular}
    }
\caption{Performance Comparison between privacy-preserving PCA algorithms in \cite{panda2021principal} and Ours under interactive modulus refresh strategy. $l_E$ denotes the desired number of principal components, $l_P$ represents the number of iterations for the PowerMethod, $l_N$ represents the average number of iterative InvSRT, $Lvs/PM$ represents the average modulus levels consumed per PowerMethod iteration, and $time$ denotes the time for the entire computation, where our approach's time consists of two parts: the computation time of the covariance matrix (left) and the computation of principal components, i.e., the PowerMethod and EigenShift algorithms (right). }     \label{tb:FinalData}
\end{table}

\subsubsection{Performance Comparison}

Based on the aforementioned experimental setup, we conduct our privacy-preserving PCA algorithm and compare its performance with the previous work of Panda. \cite{panda2021principal}. Table \ref{tb:FinalData} presents the performance of our privacy-preserving PCA algorithm and the algorithm proposed in \cite{panda2021principal}., both using the interactive modulus refresh strategy.  \\ % Table \ref{tb:FinalData2} specifically showcases the performance of our privacy-preserving PCA algorithm using the non-interactive modulus refresh strategy, which was not considered in the work by Panda et al. 
%Based on the experimental setup described, we conducted our privacy-preserving PCA algorithm and compared its performance with the previous work of Panda et al. \cite{panda2021principal}. Table \ref{tb:FinalData} 中展示了我们的以及 \cite{panda2021principal} 的、使用交互式模数刷新策略的 privacy-preserving PCA 算法的表现。Table \ref{tb:FinalData2} 单独展示了我们的、使用非交互式模数刷新策略的 privacy-preserving PCA 算法的表现，这种策略是 \cite{panda2021principal} 所没有考虑的。

% \noindent \textbf{Interactive Modulus Refresh Strategy}
\paragraph{Efficiency and Scalability}
% \noindent \textbf{Efficiency and Scalability}

We can observe that for the same number of principal components and PowerMethod iterations, our algorithm is approximately 1.9 times faster in terms of the overall runtime. Specifically, our PowerMethod runs at a speed of approximately 8.2 seconds per iteration, which is around 2.3 times faster than the 19.6 seconds reported by Panda per iteration. This notable improvement can be attributed to a lower number of ciphertext rotations achieved by our PowerMethod (see Theorem \ref{thm:PMCplx}) compared to their reported complexity of $O(l_P\cdot (2\log{d}\cdot \frac{s\cdot d}{N/2} + \log{\frac{N/2}{d}}))$ (where $s$ remains the number of samples and $d$ the number of features). This advantage becomes more apparent when dealing with larger sample sizes. As an example, we are able to compute the top 8 eigenvectors of a dataset with 60,000 samples in approximately 1 hour using 7 threads, which is estimated to be 8 times faster than the work by Panda using the same number of threads. Such scalability is not reported in previous research. \\ %因此 Our privacy-preserving PCA algorithm 同时还展现出更好的可扩展性：我们使用 7 线程在大约 1 小时内计算出具有60000个样本的数据集的 8 个特征向量，我们估计这大约是 Panda \textit{et al.} 的方案使用相同线程数时10倍的速度。

%我们的 PowerMethod 的密文旋转次数 （see Theorem \ref{thm:PMCplx}）要优于他们报告的密文旋转次数 $O(l_P\cdot (2log(d)\frac{s\cdot d}{\text{Slots}} + \log{\frac{Slots}{d}})$（其中 $s$ 是样本数量，$d$ 是特征数），在数据集样本数量较大时。

\paragraph{Accuracy and Practicality}

Our privacy-preserving PCA demonstrates an average improvement in $R2(X)$ accuracy of $0.05$ and $R2(V)$ accuracy of $0.195$ compared to the results reported by Panda for the same number of eigenvectors and PowerMethod iterations. This is attributed to the higher precision of our vector normalization strategy (Lazy Normalization) compared to theirs. They simply fix the evaluation interval of their iterative InvSRT to $[0.001,750]$, use a linear initial approximation function $y = -0.00019703x+0.14777278$, and employ 6 iterations of iterative InvSRT in each PowerMethod iteration. If we evaluate the effectiveness of their vector normalization strategy using the PowerMethod Simulation we employ in Lazy Normalization, their strategy only accommodates $4$ PowerMethod iterations for an error bound of $1e-2$ involving covariance matrix transformations with a contraction coefficient of $1$. Furthermore, their approach requires pre-scaling the datasets to match the evaluation interval of the vector normalization, posing a high probability of overflow (or underflow) when pre-scaling is done inappropriately by users. In contrast, our strategy determines the evaluation interval based on the upper bound of the dataset elements, eliminating the risk of such failures and demonstrating better practicality.\\ 

%更不理想的是，他们需要用户事先将数据集缩放至与求值区间相匹配的规模才能正确地执行 privacy-preserving PCA 算法，这就存在数据集事先未得到恰当缩放，中间值溢出而导致算法执行失败的风险。而我们的策略是根据数据集中元素的上界来确定求值区间，杜绝了这种失败的风险，显然具有更好的实用性。

%Our privacy-preserving PCA shows an average improvement of $0.18$ in terms of the $R2(X)$ score compared to the results reported by Panda et al. for the same number of eigenvectors and PowerMethod iterations. 这是由于我们的向量规范化策略（Lazy Normalization）要比他们有更高的精确度。Panda 等人的工作缺乏像我们这样系统的参数化向量规范化操作的策略，他们简单地将他们的 iterative InvSRT 的求值区间固定为 $[0.001,750]$、initial approximation function 固定为线性函数 $y = -0.00019703x+0.14777278$, 并在每轮 PowerMethod 迭代中设置 6 轮 iterative InvSRT 迭代。如果用我们在 Lazy Normalization 中使用的、对向量规范化有效性进行衡量的方法来衡量他们的向量规范化策略，那么他们的策略 only accommodates $4$ PowerMethod iterations that involve covariance matrix transformations with a contraction coefficient of $1$ with an error bound of $1e-2$.

%However, we observe that the $R2(V)$ scores vary across different datasets, even with the same number of PowerMethod iterations. This variation is due to the convergence speed of the PowerMethod, which depends on the proximity of the dominant eigenvalues of the covariance matrix. When the dominant eigenvalues are close in value, the iteration speed of the PowerMethod tends to slow down, requiring more iterations to obtain more accurate eigenvectors.
%我们的 $R2(X)$ 分数比Panda \textit{et al} 报告的提升了平均 $0.18$，在相同的目标特征向量数、PowerMethod 迭代数上。然而，从 $R2(V)$ 中可以看出，相同的 PowerMethod 迭代次数在不同的数据集上有不同的表现，这是由于 PowerMethod 的迭代速度取决于协方差矩阵的 dominant 特征值之间的接近度，如果两个 dominant 特征值相差很小的时候，PowerMethod 的迭代速度就会相对变慢，更多的 PowerMethod 需要被投入以获得更准确的特征向量。

\paragraph{Modulus Utilization}

Additionally, our modulus utilization is higher than the approach proposed by Panda. For a total of 4 PowerMethod iterations, our approach requires an average of 61 modulus levels for each PowerMethod iteration with an average of 14.5 iterative InvSRT iterations, while their approach requires an average of 77 modulus levels to complete one PowerMethod iteration with only 6 iterative InvSRT iterations. This discrepancy arises due to two factors: (i) We composite two iterative InvSRT iterations into one and employ more ciphertext multiplications to save 1 modulus level per iterative InvSRT iteration. (ii) their PowerMethod implementation has a multiplication depth that is dependent on the dataset's number of features. In contrast, our approach's multiplication depth is independent of the dataset's features. As a result, we have more flexibility in performing additional rounds of the iterative InvSRT algorithm to achieve a more accurate vector normalization. %, thereby obtaining results closer to those computed by the un-encrypted routine (i.e. having a much smaller $error$). \\

\section{Conlusion and further discussion}
In this paper, we have presented a novel and efficient privacy-preserving PCA scheme using homomorphic in the context of cloud computing, which has addressed several challenges presented in previous approaches. Firstly, we have successfully tackled the obstacle of homomorphically computing the covariance matrix, a limitation in previous methods. This is achieved by designing an efficient homomorphic covariance calculation algorithm, leveraging our optimized matrix multiplication as a core component, and utilizing parallel computations of multiple matrix multiplication instances to enhance speed.  Secondly, we have introduced an efficient homomorphic circuit for the PowerMethod algorithm, which incorporates a universal vector normalization strategy to address the issue of potential accuracy loss. The experimental results demonstrate that our approach surpasses state-of-the-art methods, offering reduced runtime, improved accuracy, and enhanced scalability.
% In conclusion, 我们首先解决了过去方案中无法同态地计算协方差矩阵的障碍，通过设计一个高效的同态协方差计算算法，其中该算法以我们的优化的矩阵乘法作为核心组件，利用了多矩阵乘法并行计算来提升速度。其次，我们为用于计算主成分的 PowerMethod 算法设计了高效的同态求值算法，其中包含了一个通用向量规范化策略及其系统的参数化方法以解决过往方案中潜在的结果不精确的问题。我们的方案对比 state-of-art 方案具有更少的运行时间，更高的准确度与更好的可扩展性。

There are several topics for further discussion, including but not limited to:
\begin{enumerate}
    \item A significant portion of the complexity in homomorphic matrix multiplication comes from row and column transformations. Recent work by Jang \textit{et al.} introduced optimizations for matrix multiplication that effectively reduce the complexity of these transformations \cite{jang2022privacy}. Their optimizations could be complementary to ours, and combining them may lead to even better performance.
    \item The performance of PowerMethod depends on the characteristics of the dataset itself, especially the distribution of eigenvalues. If the eigenvalues of the covariance matrix are close to each other, the accuracy of PowerMethod may be reduced. It may be worth considering alternative PCA algorithms to achieve more accurate results. Notably, our homomorphic covariance matrix computation algorithm can serve as a black box function for alternative privacy-preserving PCA algorithm design. %后续研究讨论可以利用我们的 homomorphic covariance matrix computation algorithm ，借助显式计算好的协方差矩阵来做 PCA 算法的进一步设计与改良。 
    \item Our proposed solution can serve as a candidate for multi-party privacy-preserving PCA. Specifically, our approach can be migrated to multi-party homomorphic encryption such as the work proposed by Mouchet \textit{et al.} \cite{mouchet2021multiparty} and provide possible solutions for both vertical and horizontal federated privacy-preserving PCA scenarios. %我们的方案可以作为多方隐私保护 PCA 的候选方案。具体来说，我们的方案可以被移植到多方同态加密上，并能支持纵向和横向的多方隐私保护 PCA。
\end{enumerate}

\begin{comment}
In conclusion, we first optimized the key space for homomorphic matrix multiplication, reducing the total number of non-zero diagonal vectors and improving the efficiency of rotation operations. We designed circuits for computing the covariance matrix using homomorphic matrix multiplication. Furthermore, we proposed a new homomorphic PowerMethod circuit that takes advantage of the symmetry of the covariance matrix, allowing the linear transformations between the eigenvectors and the covariance matrix to alternate along different axes. We also provided an analysis of the normalization process and introduced a lazy normalization method using the InvSRT iteration algorithm, which offers flexibility and adaptability in PCA parameter deployment.  
\end{comment}

\section{Acknowledgements}
This research is supported by the National Key R\&D (Research and Development) Program of China 2021YFF0704102. 

\bibliographystyle{alpha}
\bibliography{sample}

\appendix
\section{Detailed Information of Full-RNS CKKS scheme}\label{apd:A}
We use $\Phi_M(X)$ to denote the $M$-th cyclotomic polynomial and denote its $N=\varphi(M)$ roots as $\zeta_1,\zeta_3,\dots,\zeta_{\varphi(M)}$, where $\zeta_i =\zeta_1^{i}$. The symbol ${R}$ represents the cyclotomic polynomial ring $\mathbb{Z}[X]/\Phi_M(X)$. Let $Q_L=q_0\times q_1\times \dots \times q_L$ be the product of prime numbers $q_i$. The elements in ring ${R}_{Q_L}$ can be uniquely represented in the RNS field as $R_{q_0}\times R_{q_1} \times \dots \times R_{q_L}\cong R_{Q_L}$. The introduction of full-RNS CKKS can be divided into the following four modules. 

%\noindent {\large \textbf{General Module}}
\subsection{General Moudle}

This module includes the parameterization for full-RNS CKKS, the encoding method to convert a complex vector to a native plaintext form, the encryption method to convert plaintext to ciphertext, as well as the key switching technique to change the decryption key associated with a specific ciphertext.
 
\noindent \textbf{Parameter Generation} $\text{Setparams}(N,b,\sigma)$ sets $M=2N$ and selects the $M$-th cyclotomic polynomial $\Phi_M(X)=X^N+1$. Two sets of prime numbers $\{q_i|i=0,\dots,L\}$ and $\{p_i|i=0,\dots,\alpha-1\}$ are chosen such that their products are at most $b$ bits. These primes satisfy ${q_i,p_i\equiv 1 \text{ mod }2N}$. Let $Q_L=\prod_{i=0}^L q_i$ and $P=\prod_{i=0}^{\alpha} p_i$. The ring $R_{Q_L}= \mathbb{Z}[X]_{Q_L}/\Phi_M(X)$ is referred to as the plaintext space of the CKKS scheme. The distribution over $R$ for generating the secret key is denoted as $R\leftarrow \chi_s$, and the error distribution over $R$ as $R\leftarrow \chi_e$, which is a truncated discrete Gaussian distribution with standard deviation $\sigma$.

\noindent \textbf{Encoding:} $\text{Encode}(\mathbf{z},\Delta,n,\ell)$ takes a message vector $\mathbf{z}\in \mathbb{C}^{n}$, where $1\leq n <N$ and $n$ divides $N$, and converts it to a plaintext polynomial $m\in R_{Q_\ell}$. To prevent a significant precision loss in the conversion, a scaling factor $\Delta$ is applied to the message. For any plaintext in $R_{Q_\ell}$, denote $\ell$ as its modulus level.

\noindent \textbf{Decoding:} $\text{Decode}(m,\Delta,n,\ell)$ takes a plaintext polynomial $m\in {R}_{Q_\ell}$ as input and converts it to a message vector $\mathbf{z}\in \mathbb{C}^{n}$, where $1\leq n <N$ and $n$ divides $N$. The decoding process is the inverse operation of encoding. 

\noindent \textbf{Secret Key Generation:} $\text{SecKeyGen}(\cdot)$ samples $s\leftarrow \chi_s$ and outputs the secret key $s$.

\noindent \textbf{Public Key Generation:} $\text{PubKeyGen}(s)$ samples $e\leftarrow \chi_e$ and $a\in_u R_{Q_L}$, and outputs $(-as+e,a)$.

\noindent \textbf{Encryption:} $\text{Enc}(m,pk)$ takes a plaintext polynomial $m\in {R}_{Q_\ell}$ and a public key $pk\in R^2_{Q_L}$ as input. It samples $a\in {R}_{Q_\ell}$  and outputs $(a\cdot pk_0,a\cdot pk_1) + (m+e_0,e_1)\in {R}^2_{Q_\ell}$ as a ciphertext encrypting $m$. For any ciphertext in ${R}^2_{Q_\ell}$, denote $\ell$ as its own modulus level.

\noindent \textbf{Decryption:} $\text{Dec}(c,sk)$ takes a ciphertext $c\in R_q^2$ and a secret key $sk$ as input and outputs a plaintext polynomial $m'\in {R}$. It computes and outputs $m' = c_0 + c_1\cdot sk\in R_{Q_\ell}$.

%Besides the public key and secret key for encryption and decryption, switching keys are required for some homomorphic evaluation routines. They are used to convert a ciphertext under secret $s$ to the ciphertext under another secret $s'$ without changing its underlying plaintext. 

\noindent \textbf{Switch Key Generation:} $\text{SwitchKeyGen}(s,s',\mathbf{b})$ for $\mathbf{b}\in \mathbb{Z}^\beta$ a vector representation of a certain integer decomposition basis, \text{SwitchKeyGen} takes in $s,s'$ and outputs a vector as the switching key: $swk_{s,s'} = (swk^{(0)}_{s,s'},\dots, swk^{(\beta-1)}_{s,s'})$ where $swk^{(i)}_{s,s'} \in R_{PQ_L}^2$. 

%\noindent \textbf{Relinearization Key Generation} $\text{RelinKeyGen}(s,s^2,\mathbf{b})$ $\mathbf{b}$ is a vector representing a certain integer decomposition basis with $\beta$ units. The algorithm computes $\text{rlk} =\left((-a_0s^2+sb^{0}P+e_0,a_0),\dots,(-a_\beta s^2+sb^{\beta}P+e_\beta,a_\beta)\right)$, where for $0\leq i \leq \beta$, $a_i$ is sampled from $R_{PQ_L}$, and $e_i\leftarrow \chi_e$.

%\textbf{Rotation Key Generation:} $\text{RotKeyGen}(s,s^{5^k},\mathbf{b})$ generates rotation keys. Here, $\mathbf{b}$ is a vector representing a certain integer decomposition basis with $\beta$ units. The algorithm computes $\text{rtk}*k =\left((-a_0s^{5^k}+sb^{0}P+e_0,a_0),\dots,(-a*\beta s^{5^k}+sb^{\beta}P+e_\beta,a_\beta)\right)$, where for $0\leq i \leq \beta$, $a_i$ is sampled from $R_{PQ_L}$, and $e_i\leftarrow \chi_e$.

\noindent \textbf{Key Switching:} $\text{KeySwitch}(c,swk_{s,s'})$ takes a polynomial $c=(a,b)\in R_{Q_\ell}^2$ and a switching key $swk_{s,s'}$ as input. It first decomposes $a$ with respect to the decomposition basis $\mathbf{b}$ of $swk_{s,s'}$, i.e., $a = \left< \mathbf{a},\mathbf{b} \right>$. Then, it computes and returns $(a_0,a_1) = \lfloor P^{-1}\cdot \left< \mathbf{a},swk_{s,s'} \right> \rceil\text{ mod } Q_\ell$. 

\subsection{Addition Module}
This part introduces the function provided by full-RNS CKKS to achieve homomorphic addition between ciphertexts, which induces the underlying plaintexts to perform Hadamard addition.

\noindent \textbf{Plaintext Addition:} $\text{PAdd}(c,m)$ takes a plaintext polynomial $m\in {R}_{Q_\ell}$ and a ciphertext $c\in {R}_{Q_\ell}^2$ as input, both having the same scaling factor. It outputs a ciphertext $c' = c+(m,0)$.

\noindent \textbf{Ciphertext Addition:} $\text{Add}(c_1,c_2)$ takes two ciphertexts $c_1,c_2 \in {R}_{Q_\ell}^2$ with the same scaling factor as input and outputs $c_1+c_2$. 

%\noindent {\large \textbf{Multiplication Module}}
\subsection{Multiplication Module}
The multiplication module contains the functions to achieve homomorphic multiplication among ciphertexts and plaintexts, which induces the underlying plaintexts to perform the Hadamard product.
% \textbf{Constant Real Multiplication:} $\text{CMult}(c,a\in \mathbb{R},\Delta')$: For a ciphertext $c\in \mathcal{R}_{Q_\ell}^2$ with scaling factor $\Delta$, compute and return $\lfloor \Delta' a \rceil\cdot c$. Here, $c$ has a scaling factor of $\Delta\Delta'$, and the underlying plaintext is multiplied by $a$.

\noindent \textbf{Plaintext Multiplication:} $\text{PMult}(c,m)$: Given a plaintext polynomial $m\in {R}_{Q_\ell}$ and a ciphertext $c\in {R}_{Q_\ell}^2$ with scaling factors $\Delta$ and $\Delta'$ respectively, compute and return $c' = (c_0\cdot m , c_1\cdot m)$. The resulting ciphertext $c'$ has a scaling factor of $\Delta\Delta'$.

\noindent \textbf{Ciphertext Multiplication:} $\text{Mult}(c_1,c_2)$: For ciphertexts $c$ and $c'$ with scaling factors $\Delta$ and $\Delta'$ respectively, compute $(d_0,d_1,d_2) = (c_0c_0',c_0c_1'+c_1c_0',c_1c_1')$ and return $d=(d_0,d_1)+\text{SwitchKey}(d_2,\text{rlk})\in {R}_{Q_\ell}^2$ as the resulting ciphertext, where $\text{rlk}$ is a switching key $swk_{s,s^2}$. The resulting ciphertext $d$ has a scaling factor of $\Delta\Delta'$.

A rescale operation is commonly used to transfer a ciphertext from a ring with a large modulus to a ring with a smaller modulus, reducing the expansive scaling factor in the ciphertext and limiting the error introduced by multiplications.

% \noindent \textbf{Level Dropping:} $\text{Drop}(c,k)$: Given a ciphertext $c\in {R}_{Q_\ell}^2$, return $c \in {R}^2_{Q_{\ell-k}}$.

\noindent \textbf{Rescaling:} $\text{Rescale}(c)$: For a ciphertext $c\in {R}_{Q_\ell}^2$ with scaling factor $\Delta$, return $\lfloor q_\ell^{-1}\cdot c \rceil$. The resulting ciphertext has a scaling factor of $\Delta/q_{\ell}$.

% \textbf{Set Scaling Factor:} $\text{SetScale}(c,\Delta')$: For a ciphertext $c\in \mathcal{R}_{Q_\ell}^2$ with scaling factor $\Delta$, return $\text{Rescale}(\text{CMult}(c,\Delta'/\Delta ,q_{\ell}))$. The resulting ciphertext has a scaling factor of $\Delta'$ and belongs to $\mathcal{R}^2_{Q_{\ell-1}}$.
% However, for a given set of modulus products $Q_L$, the homomorphic operations are limited to a multiplication depth of at most $L$ because rescaling consumes the modulus level of ciphertexts. Once the modulus level of a ciphertext reaches 0, it can no longer participate in homomorphic multiplication operations. To perform operations with a greater multiplication depth, we need to conduct an operation called modulus refresh which raises the modulus level of the ciphertext. The modulus refresh can be achieved either by decrypting and re-encrypting the ciphertexts or by applying the bootstrapping function to the ciphertexts. %However, 对于给定的一组模数乘积 $Q_L$，CKKS 只能支持乘法深度至多为 $L$ 的同态运算，因为 Rescaling 需要消耗的密文的模式层级是有限的，一旦一个密文的模式层级降至 0，那么它无法再参与同态乘法运算。要想进一步进行乘法深度更大的运算，我们需要将低模数层级的密文刷新至较高的模数层级，这可以通过解密，重新加密实现，也可以通过对密文应用自举函数实现。
% However, since the modulus is finite, once it is depleted, the scheme can no longer use rescaling. Therefore, we need techniques to bring the ciphertext back to a ring with a larger modulus, which is the purpose of bootstrapping \cite{cheon2018bootstrapping}.

\noindent \textbf{Bootstrapping:} $\text{Bootstrap}(c)$: Bring a ciphertext $c\in R^2_{Q_\ell}$ back to $R^2_{Q_{L-k}}$, where $k$ is the multiplication depth of the bootstrapping circuit.  

%We refer to the bootstrapping operation and the process of decrypting and re-encrypting (re-encryption) as modulus refresh operations, as they both support transferring the ciphertext from a lower-level modulus ring to a higher-level modulus ring. \\ %We refer bootstrapping and re-encryption operations collectively as \textbf{modulus refresh} operations, as they both bring the ciphertext from a lower-level modulus ring to a higher-level modulus ring. 

%\noindent {\large \textbf{Rotation Module}}
\subsection{Rotation Module}
The rotation module includes the functions to achieve homomorphic position shifting in a ciphertext. Specifically, full-RNS CKKS allows us to homomorphically rotate the position of the components of the underlying plain complex vector in a ciphertext. 

\noindent \textbf{Rotation:} $\text{Rot}(c_1,k)$: Given a ciphertext $c\in {R}_{Q_\ell}^2$ and a rotation key $\text{rtk}_k$, compute and output $(c_0^{5^k},0)+\text{SwitchKey}(c_1^{5^k},\text{rot}_k)$, where $\text{rot}$ is a switching key $swk_{s,s^{5^k}}$. This ciphertext represents the result of rotating all components of the message vector $\mathbf{z}\in \mathbb{C}^{N/2}$ encrypted by $c$ to the left by $i$ steps (or positions), denoted as $\mathbf{z}' = \rho(\mathbf{z}; i)$. \\

\section{Proofs for the Vector Normalization Strategy} \label{apd:B}

\subsection{Amplification Upper Bound of the Shifted Covariance}

\begin{theorem}\label{thm: EigShiftnoBig}
    For the $1$-shifted covariance variance matrix $\Sigma' = \Sigma - \lambda\mathbf{v}\mathbf{v}^T$ of the covariance matrix $\Sigma$, where $\mathbf{v}$ is the normalized approximate dominant eigenvector of $\Sigma$ obtained using the PowerMethod, and $\lambda=\frac{\left<\Sigma\mathbf{v},\mathbf{v}\right>}{\left<\mathbf{v},\mathbf{v} \right>}$, the matrices $\Sigma'$ and $\Sigma$ have the same scaling upper bound. In other words, the amplification effect of the transformation $\Sigma'$ on a vector does not exceed the amplification effect of the transformation $\Sigma$ on the same vector. %对于协方差矩阵 $\text{Cov}$ 的 $1$-shifted 协方差方差矩阵 $\text{Cov'} = \text{Cov} - \lambda\mathbf{v}\mathbf{v}^T$，其中 $\mathbf{v}$ 是用 PowerMethod 求得的 $\text{Cov}$ 的规范化的近似 dominant eigenvector，$\lambda=\frac{\left<\text{Cov}\mathbf{v},\mathbf{v}\right>}{\left<\mathbf{v},\mathbf{v} \right>}$，那么 $\text{Cov}'$ 和 $\text{Cov}$ 具有相同缩放上界，即 $\text{Cov}'$ 变换对一个向量的放大效果不会超过 $\text{Cov}$ 变换对相同向量的放大效果。 
\end{theorem}
\begin{proof}
    Let the covariance matrix have $n$ eigenvectors: $x_1,x_2,\dots,x_n$, corresponding to eigenvalues $\lambda_1,\lambda_2,\dots,\lambda_n$. We only need to prove that $\lambda \leq \lambda_1$. When $\mathbf{v}$ is equal to $\mathbf{x}_1$, the inequality is obviously true. When $\mathbf{v}$ is not equal to $\mathbf{x}_1$, it can be represented as $\mathbf{v} = c_1\mathbf{x}_1+c_2\mathbf{x}_2+\dots+c_n\mathbf{x}_n$, where $c_1^2 + c_2^2+\dots +c_n^2 = 1$. In this case, $\left<\Sigma\mathbf{v},\mathbf{v}\right> = \lambda_1c_1^2+\lambda_2c_2^2+\dots +\lambda_nc_n^2 \leq \lambda_1(c_1^2+\dots +c_n^2) = \lambda_1$, while $\left<\mathbf{v},\mathbf{v} \right>$ is always equal to 1. Therefore, $\lambda \leq \lambda_1$ holds true. %设协方差矩阵有 $n$ 个特征向量：$x_1,x_2,\dots,x_n$，对应特征值 $\lambda_1,\lambda_2,\dots,\lambda_n$，那么我们只需要证明 $\lambda \leq \lambda_1$ 即可。当 $\mathbf{v}$ 即为 $\mathbf{x}_1$ 时，不等式显然成立。当 $\mathbf{v}$ 不为 $\mathbf{x}_1$ 时，它一定能被表示成 $\mathbf{v} = c_1\mathbf{x}_1+c_2\mathbf{x}_2+\dots+c_n\mathbf{x}_n$，其中 $c_1^2 + c_2^2+\dots c_n^2 = 1$。那么此时 $\left<\text{Cov}\mathbf{v},\mathbf{v}\right> = \lambda_1c_1^2+\lambda_2c_2^2+\dots +\lambda_nc_n^2 \leq lamda_1(c_1^2+\dots +c_n^2) = lamda_1$，而 $\left<\mathbf{v},\mathbf{v} \right>$ 恒为 1，因此 $\lamda \leq \lambda_1$ 成立。
\end{proof}
\subsection{Proof of Effectiveness for the PowerMethod Simulation}
\begin{theorem}\label{thm: ContractionCoeff}
    If a specific instance of the Lazy Normalization strategy can maintain an error bound of $e$ under a PowerMethod simulation with a Contraction Coefficient of $S$ for $t$ iterations, then for any covariance matrix with eigenvalues all greater than $S$, this instance guarantees that the length of the output approximate eigenvector within $t$ PowerMethod iterations is close to $1$ with an error no bigger than $e$.
\end{theorem}
\begin{proof}
    We need to prove that the scaling effect of the covariance matrix transformation on any vector $\mathbf{v}$ with a length of 1 is not smaller than its minimum eigenvalue $\lambda_n$. Let $\lambda_1,\lambda_2,\dots,\lambda_n$ represent the eigenvalues of $\Sigma$ arranged in descending order, and $\mathbf{x}_1,\dots,\mathbf{x}_n$ be their corresponding eigenvectors. It is known that $\mathbf{v}$ can be expressed as: $\mathbf{v} = c_1\mathbf{x}_1+c_2\mathbf{x}_2+\dots+c_n\mathbf{x}_n$, where $c_1^2 + c_2^2+\dots+c_n^2 = 1$. Then $||\Sigma\mathbf{v}||_2^2 = \lambda_1^2c_1^2+\dots+\lambda_n^2c_n^2 \geq \lambda_n^2(c_1^2+\dots+c_n^2)=\lambda_n$.
    %我们只需要证明协方差矩阵变换对任意具有长度 1 的向量 $\mathbf{v}$ 的缩放效果均不小于它的最小特征值 $\lambda$ 即可。设 $\lamda_1,\lambda_2,\dots,\lambda_n$ 代表 $\Sigma$ 的从大到小排列的特征值，$\mathbf{x}_1,\dots,\mathbf{x}_n$ 是它们对应的特征向量。已知 $\mathbf{v}$ 能被表示成：$\mathbf{v} = c_1\mathbf{x}_1+c_2\mathbf{x}_2+\dots+c_n\mathbf{x}_n$，其中 ，$c_1^2 + c_2^2+\dots c_n^2 = 1$。那么 $||\Sigma\mathbf{v}||_2^2 = \lambda_1^2c_1^2+\dots +\lamda_n^2c_n^2 \geq \lambda_n^2(c_1^2+\dots+c_n^2)=\lambda_n^2$。
\end{proof}

\section{Modulus Refresh Timing in the Privacy-preserving PCA algorithm} \label{apd:C}
\begin{figure}[htbp]
\centerline{\includegraphics[scale=0.5]{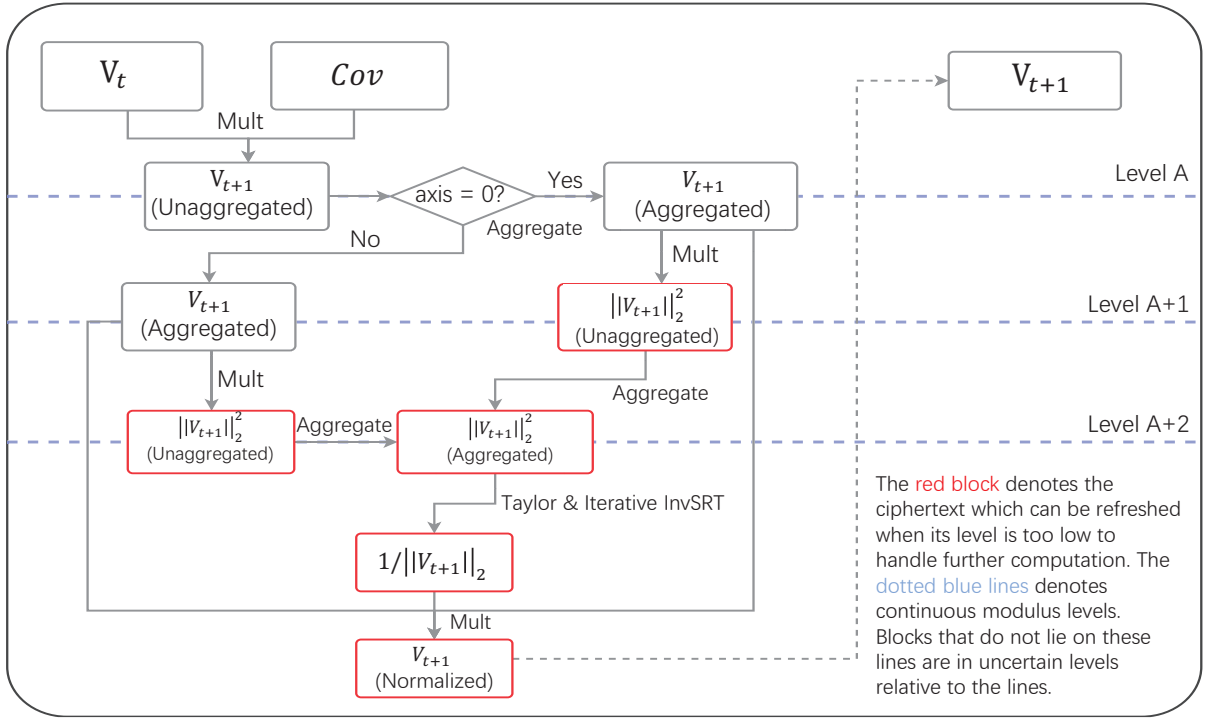}} %%图像路径：pic是文件夹名。
\caption{Flow chart of the modulus refresh strategy in homomorphic PowerMethod }
\label{fig:flowchart}
\end{figure}

%One advantage of our privacy-preserving PCA scheme is that 
No modulus refresh operations are necessary during the computation of the covariance matrix, as its multiplication depth can be handled easily within concrete CKKS schemes. Regarding the PowerMethod, the timing of modulus refresh operations is illustrated in Flowchart \ref{fig:flowchart}. As depicted in the flowchart, we avoid performing modulus refresh during the covariance matrix transformation because doing so would increase the frequency of modulus refreshes proportional to the size of the covariance matrix. Instead, we only refresh some intermediate values or the approximate eigenvector. Specifically, we manage to perform modulus refresh only once when refreshing the approximate eigenvector. This is achieved by extracting the components of the approximate eigenvector distributed across different ciphertexts and encoding them into a single ciphertext. After the refresh, the components can be re-distributed into different ciphertexts. This strategy effectively decouples the number of modulus refreshes from the number of vector components.    
In addition, it is worth noting that the cloud users in our scheme are not involved in any aspect of the privacy-preserving PCA algorithm, apart from encryption and decryption operations. Unlike previous approaches that require user participation in the privacy-preserving algorithmic process \cite{rathee2018faster}, \cite{lu2016using}, our scheme not only safeguards the algorithm privacy of the cloud service provider but also minimizes the coupling between the privacy-preserving PCA algorithm and the design of the communication protocol between the cloud and the user. % Regardless of the chosen modulus-level refresh strategy, 在我们方案中，除了加解密以外，云服务商的用户不参与隐私保护 PCA 算法中的任意环节，对比过往的一些需要用户参与算法过程的方案，我们的方案一方面保护了云服务商的算法权益，令一方面将隐私保护 PCA 算法与云服务通信协议设计之间的耦合度降至最低。

\end{document}